\documentclass[review]{elsarticle}

\usepackage{amsmath,amsthm,amsfonts,amssymb}
\usepackage{vhistory}

\usepackage{caption}
\usepackage[left=1.1in,right=1.1in,top=1in,bottom=1in]{geometry}

\newcommand{\R}{{\mathbb{R}}}

\newcommand{\Q}{{\mathbb{Q}}}

\newcommand{\cC}{\mathcal{C}}

\newcommand{\cF}{\mathcal{F}}

\newcommand{\cZ}{\mathcal{Z}}

\newtheorem{definition}{Definition}[section]

\newtheorem{proposition}[definition]{Proposition}
\newtheorem{lemma}[definition]{Lemma}
\newtheorem{corollary}[definition]{Corollary}
\theoremstyle{definition} \newtheorem{remark}[definition]{Remark}
\theoremstyle{definition} \newtheorem*{remark*}{Remark}

\usepackage{algpseudocode}
\usepackage[ruled,vlined]{algorithm2e}

\begin{document}

\begin{frontmatter}

\title{Calibration of Local Volatility Model with Stochastic Interest Rates by Efficient Numerical PDE Method}

\author{Julien Hok \footnote{Corresponding author}}
\address{Credit Agricole CIB, Broadwalk House, 5 Appold St\\
		 London, EC2A 2DA, United Kingdom\\
		\it{julienhok@yahoo.fr}}

\author{Shih-Hau Tan \footnote{Quantitative analyst consultant}}
\address{Cuemacro, London, United Kingdom} 

\begin{versionhistory}
  \vhEntry{1.0}{11.03.2018}{}{1st version}
\end{versionhistory}

\begin{abstract}
Long maturity options or a wide class of hybrid products are evaluated using a local volatility type modelling for the asset price $S(t)$ with a stochastic interest rate $r(t)$.  The calibration of the local volatility function is usually time-consuming because of the multi-dimensional nature of the problem.  In this paper, we develop a calibration technique based on a partial differential equation (PDE) approach which allows an efficient implementation.  The essential idea is based on solving the derived forward equation satisfied by    
$P(t, S, r) \cZ (t, S, r)$, where $P(t, S, r)$ represents the risk neutral probability density of $(S(t), r(t))$ and  $\cZ (t, S, r)$ the projection of the stochastic discounting factor in the state variables $(S(t), r(t))$.
The solution provides effective and sufficient information for the calibration and pricing.  The PDE solver is constructed by using ADI (Alternative
Direction Implicit) method based on an extension of the Peaceman-Rachford scheme. 
Furthermore, an efficient algorithm to compute all the corrective terms in the local volatility function due to the stochastic interest rates is proposed by using the PDE solutions and grid points.  Different numerical experiments are examined and compared to demonstrate the results of our theoretical analysis.
\end{abstract}


\end{frontmatter}

\section{Introduction}


In quantitative finance, the local volatility type model as introduced in \cite{DerKa98,Dupire94,Rubi94} is widely used to model the price of underlying in order to capture the market volatility skew or smile in equity or foreign exchange market.  It is known that with deterministic rates, the local volatility function can be obtained with the Dupire formula (see equation (\ref{LVDetIR})) by using the European call and put option prices. \\

For long maturity options (e.g pure equity autocall derivative) or some hybrid products 
with a payoff involving interest rate and underlying asset like a best-of interest rate-equity which pays coupons of the form 
\begin{equation}
\max \left[ LIBOR, a \left(\frac{S_t}{S_0}-1\right) \right],
\end{equation}  
where the interest rates is potentially needed to be modelled as stochastic. It is then natural to extend the local volatility model to incorporate stochastic interest rates. This modeling framework is widely used in the financial industry (see e.g \cite{Atlan06,Clark11,GobetBenMiri12,MarcusBermudezBuehler07}). To perfectly match the market implied volatility, the local volatility type formula can be derived and is given by equation (\ref{LVStoIR}).  We observe on top of the Dupire local volatility function, there is an additional corrective term taking into account the covariance between the equity price and short rate.  Unfortunately, this extension of the {\it{Dupire}} formula is not easily applicable for calibration over the market since there seems no immediate way to link the expectation term with the European option prices or other liquid products.  \\

The main challenge for the implementation of the model consists in the calibration (see also discussion in \cite{Piterbarg06}). 
First and foremost, it corresponds to a two factor models. Also the formula (\ref{LVStoIR}) requires a correction term
on top of {\it{Dupire}} expression for each strike point $K$ per maturity in the definition of the local volatility function.
Accuracy, robustness in the   calibration and pricing associated with an efficient implementation are required for execution in real time.  Indeed, the model will be used not only for the pricing but also to compute all sensitivity factors for hedging purpose (e.g delta, gamma and vega).\\

Many research results appeared in the last decade in different areas of quantitative finance on the local volatility model with stochastic rates.  A brief introduction is given as following:

\begin{itemize}

\item Theoretical results about the local volatility function and its calibration were exposed and discussed in 
\cite{Atlan06,DeelstraRayee12,RenMadanQian07}.

\item For the option pricing, in \cite{GobetBenMiri12,GobetHok14}, the authors developed expansion formulas by applying the perturbation method using a proxy introduced by \cite{BenGobetMiri09}. A pricing framework via Partial Differential Equation 
(PDE) approach was studied in \cite{DangChristaraJacLak12} and the Crank-Nicolson scheme also
the Alternating Direction Implicit (ADI) method were applied to build the PDE solver.

\item In terms of model calibration, 
for pricing the Power Reverse Dual Currency (PRDC) derivatives, Piterbarg in \cite{Piterbarg06} modeled the local volatility function for the forward foreign exchange rate using the constant elasticity variance (CEV) dynamic as a parametric form.  A fast calibration procedure based on the 
so-called Markovian projection method was developed and the skew averaging technique were discussed in \cite{Piterbarg05c,Piterbarg05d}.  The calibration essentially captures the slope of the implied volatility surface but does not exactly fit its convexity. In \cite{MarcusBermudezBuehler07}, the authors proposed a PDE calibration method which bootstraps the local volatility function for each expiry $T_i$ by solving a forward PDE for the joint distribution of $(S, r)$ under the forward measure $Q_{T_i}$.  Depending on the number of expiries in the local volatility calibration, the algorithm is potentially very intensive in computations. Calibration by Monte Carlo approach using McKean’s particle method was studied in \cite{GuyonLabordere11}. The authors in \cite{BenGruz08} used Malliavin calculus to derive an equation of the local volatility function, 
which is then solved by using fixed-point method.  However the calibration quality     
deteriorates significantly for far out of the money or long maturities as discussed in \cite{GuyonLabordere11}.

\item  For models calibration in quantitative finance, different numerical solvers were constructed to solve similar types of high dimensional PDEs.  For example, an idea of using operator splitting method was provided by Ren \cite{ren2007calibrating} for the calibration of local volatility with stochastic volatility.  A Heston-like term-structure model in FX market was examined by Tian \cite{tian2015calibrating} with a modified Douglas scheme and Wyns \cite{wyns2017adjoint} with a modified Craig-Sneyd scheme.  These ADI-type schemes were discussed in detailed for dealing with convection-diffusion equations in terms of the stability and second order convergence in \cite{wyns2016convergence}.

\end{itemize}

In this paper, we assume a Markovian setting with a stochastic differential equation (SDE) given by (\ref{OriginalSDEs}), i.e., the underlying $S$ follows a local volatility type diffusion and the short rate $r(t)$ is represented by a general form dynamic.  This model covers a particular case, a local volatility diffusion for $S$ associated with a Gaussian Hull-White dynamic for $r$ (see \cite{HullWhite93}), which is widely used for pricing hybrid products (see e.g \cite{GobetBenMiri12,GobetHok14,MarcusBermudezBuehler07}).   We focus on the model calibration and propose a methodology for an efficient implementation using a PDE approach.  The idea starts with introducing the projection of the stochastic discounting factor in the state variable $(S(t), r(t))$, quantity defined by the functional $\mathcal{Z}(t, S(t), r(t))$ in equation (\ref{zetafunction}).  Using a martingale 
argument, we derive the forward partial differential equation satisfied by the product $P(t, S, r) \cZ (t, S, r)$, where $P(t, S, r)$ is the probability density of $(S(t), r(t))$ (see our main result in proposition \ref{FPDEforPZ}).  The solution of the PDE can substantially help us to obtain the corrective terms and perform the pricing.  For getting the numerical PDE solution efficiently, a simple and intuitive ADI scheme based on an extension of the Peaceman-Rachford scheme in \cite{PeacemanRachford55} is proposed which leads to solve the discretized linear system with a reasonable matrix size.  Moreover, we suggest an efficient algorithm to compute all the corrective terms sequentially in strike per maturity using the solutions and points in the PDE grid (see section \ref{EfficientImplementation}). Generally in the two dimensional case, we believe our methodology 
is able to make PDE approach very efficient w.r.t Monte Carlo method (see discussion in \cite{GuyonLabordere11}).\\

The paper is organized as follows. Section 2 describes the hybrid equity-interest rates modelling and the derivation of the local volatility function is explained. The calibration framework is described in section 3 to
obtain the forward equation satisfied by $P(t, S, r) \cZ (t, S, r)$. For its resolution, an ADI-type method extending Peaceman-Rachford scheme is constructed in section 4. Section 5 is devoted to the numerical tests. Finally the conclusions and discussions are given in section 6.

\section{Hybrid equity interest rates model}\label{sec:Hybrid equity interest rates model}

Let's consider the 2-d stochastic differential equations, describing the spot price $S(t)$ and short rates $r(t)$ under the risk neutral probability $\Q$, defined by

\begin{equation} \label{OriginalSDEs}
\left\lbrace
\begin{array}{l}
	\frac{dS(t)}{S(t)} = r(t)dt + \sigma(t, S(t))dW^1(t) , \hspace{0.3cm} S(0) = S_0 ,\\
	dr(t) = \mu(t,r(t))dt + \alpha(t,r(t)) (\rho dW^1(t) + \sqrt{1-\rho^2} dW^2(t)), \hspace{0.3cm} r(0) = r_0, \\
	\end{array}
\right.
\end{equation}
where $(W(t))_{t \geq 0}$ is a standard Brownian motion in $\R^2$ on a filtered probability space
$(\Omega, \cF, (\cF_t)_{t \geq 0}, \Q)$ with the usual assumptions on the filtration $(\cF(t))_{t \geq 0}$. We assume the existence and uniqueness of the solution for (\ref{OriginalSDEs}) (see e.g theorem 3.5.5 in \cite{LambertonLapeyre12}). \\


This model corresponds to an extension of the {\it{Dupire local volatility}} model 
(e.g \cite{DerKa98,Dupire94,Rubi94}) which allows stochastic interest rates.
The local volatility function $\sigma(t,S)$  allows the model to calibrate to the surface of European call price 
$C(T, K)$ where $K$ represents the strike and $T$ is the maturity. For the sake of completeness, we provide its derivation by following the methodology in \cite{MusielaRut05}. Let's note by 

\begin{equation} \label{discountingfactor}
Z(t) := e^{-\int_0^t r(u)du},
\end{equation}
and applying Tanaka's formula to the convex but non-differentiable function $Z(t)(S(t)-K)_+$ leads to 

{\small{
\begin{align}
Z(t)(S(t)-K)_+ & =  (S(0)-K)_+ - \int_0^T r(u) Z(u)(S(u)-K)_+du  +\int_0^T Z(u) 1_{ S(u) \geq K }dS(u) +
				 \frac{1}{2} \int_0^T Z(u)dL_u^K(S),
\end{align}   
}}
where $L_u^K(S)$ is the local time of $S$, and $(S(0)-K)_+ = \max(S(0)-K,0)$. Since $S$ is a continuous semimartingale, then 
almost-surely (see e.g \cite{RevuzYor01})

\begin{equation}
L(t)^K(S) = \lim_{\epsilon \searrow 0} \frac{1}{\epsilon} \int_0^t 1_{[K, K+\epsilon](S(u))d<S,S>_u}.
\end{equation}

Using (\ref{OriginalSDEs}), it comes

\begin{align}
Z(t)(S(t)-K)_+ & =  (S(0)-K)_+ - \int_0^T r(u) Z(u)(S(u)-K)_+du  +\int_0^T Z(u) 1_{ S(u) \geq K }r(u)S(u)du 
				 \\	\nonumber
			&	+\int_0^T Z(u) 1_{ S(u) \geq K }S(u) \sigma(u, S(u))dW^1_u +\frac{1}{2} \int_0^T Z(u)dL_u^K(S)\\ \nonumber
			&    = (S(0)-K)_+ + \int_0^T K Z(u) 1_{ S(u) \geq K }r(u)du + \int_0^T Z(u) 1_{ S(u) \geq K }S(u) 				\sigma(u, S(u))dW^1_u \\ \nonumber
			 &    + \frac{1}{2} \int_0^T Z(u)dL_u^K(S) \\ \nonumber
			&    = (S(0)-K)_+ + \int_0^T K Z(u) 1_{ S(u) \geq K }(r(u)-f(0,u))du + \int_0^T K Z(u) 1_{ S(u) \geq K }f(0,u)du \\ 
			&+ \int_0^T Z(u) 1_{ S(u) \geq K }S(u)	\sigma(u, S(u))dW^1_u + \frac{1}{2} \int_0^T Z(u)dL_u^K(S), \label{introducefwd}
\end{align}
where $f(0,u) = -\frac{\partial}{\partial u} \log(ZC(0, u))$ means the forward rate at time $0$ for investing at time $u$ and $ZC(t, T)$ is the zero coupon price at $t$ for maturity $T$. We introduce the forward rate in the last equality in order to compare below the expression of the local volatility when interest rate is stochastic, $\sigma^2(T,K)$, and when interest rate is deterministic $\sigma^2_{Dup}(T,K)$.\\

Assuming that the function $Z(u) 1_{ S(u) \geq K }S(u) \sigma(u, S(u))$ is a member of the class $\mathcal{V}$ (see def 3.1.4 in \cite{Oksendal03}), namely the measurable and adapted functions $f$ s.t $E[\int_0^t f^2(s)ds] < \infty$,  and by taking the expectation

\begin{align}
C(T,K) & =  C(0,K) + K \int_0^T E[ Z(u) 1_{ S(u) \geq K }(r(u)-f(0,u)) ] du + K \int_0^T E[ Z(u) 1_{ S(u) \geq K }f(0,u)]du \\ \nonumber
&  + \frac{1}{2} K^2 \int_0^T E[Z(u) \delta(S(u)-K)\sigma^2(u,K)],
\end{align}
with $\delta(.)$ the delta function at $0$. Differentiating w.r.t T leads to

 \begin{align}\label{CT}
C_T(T,K) & =  K E[ Z(T) 1_{ S(T) \geq K }(r(T)-f(0,T)) ] + K E[ Z(T) 1_{ S(T) \geq K }f(0,T)] \\ \nonumber
&  + \frac{1}{2} K^2 E[Z(T) \delta(S(T)-K)\sigma^2(T,K)].
\end{align}

Standard computations give

\begin{align}
E[Z(T)1_{ S(T) \geq K }] = &-\frac{\partial C(T,K)}{\partial K}, \\
E[Z(T) \delta(S(T)-K)] = & \frac{\partial^2 C(T,K)}{\partial K^2}.
\end{align} 

Using the last two equations in (\ref{CT}), we obtain the expression of the local volatility $\sigma^2(T,K)$
in terms of call prices $C(T, K)$

\begin{equation}\label{LVStoIR}
\sigma^2(T,K) = \sigma^2_{Dup}(T,K) - \frac{E[Z(T) (r(T) - f(0, T))1_{S(T) > K}]}{\frac{1}{2} K \frac{\partial^2 C(T,K)}{\partial K^2}},
\end{equation}
with 

\begin{equation}\label{LVDetIR}
\sigma^2_{Dup}(T,K) = \frac{ \frac{\partial C(T,K)}{\partial T} + Kf(0, T)\frac{\partial C(T,K)}{\partial K} }{\frac{1}{2} K^2 \frac{\partial^2 C(T,K)}{\partial K^2}}.
\end{equation}

$\sigma_{Dup}(T,K)$ represents the {\it{Dupire}} local volatility function when interest rates are deterministic.
Equation (\ref{LVStoIR}) shows the corrections to employ on the tractable {\it{Dupire}} local
volatility surface in order to obtain the local volatility surface which takes into account the effect of
stochastic interest rates. $E[Z(T) (r(T) - f(0, T))1_{S(T) > K}]$ represents the numerator of the extra term in the local volatility expression. No closed form solution exists and it is not directly related to European call prices or other liquid products. Its calculations need to be estimated. For sanity check, when interest rate becomes deterministic, we have $r(T) = f(0, T)$ (see equation (\ref{rTQT})) and  $\sigma^2(T,K)$ reduces to $\sigma^2_{Dup}(T,K)$ in equation (\ref{LVStoIR}).

\remark{ Under $T-$forward measure $\Q^T$, the extra term can be written as

\begin{equation}\label{extratermQT}
E[Z(T) (r(T) - f(0, T))1_{S(T) > K}] = ZC(0,T)E^T[(r(T) - f(0, T))1_{S(T) > K}].
\end{equation}

In the Heath–Jarrow–Morton (HJM) framework under $\Q^T$ (see e.g \cite{BrigoMercurio06}), the forward rate $f(t,T)$ is a martingale, with a vector of volatility $\sigma(s,T)$ which can be written as

 \begin{equation}
f(t,T) = f(0,T) + \int_0^t \sigma(s,T).dW^{T}_s.
\end{equation}

Using the fact that $r(T) = f(T,T)$, it becomes

 \begin{equation}\label{rTQT}
r(T) = f(0,T) + \int_0^T \sigma(s,T).dW^{T}_s,
\end{equation}
and with (\ref{rTQT}), the expression of the extra term in (\ref{extratermQT}) shows clearly the impact of stochastic rates. \\

By assuming $\int_0^T \sigma^2(s,T)ds < +\infty$, $\int_0^T \sigma(s,T).dW^{T}_s$ is a $\Q^T$ martingale and 

\begin{equation}
E^T[(r(T) - f(0, T))1_{S(T) > K}] = Cov^T[(r(T) - f(0, T)), 1_{S(T) > K}],
\end{equation}
which provides an interpretation to the corrective term as a covariance, under $\Q^T$,  between  
$r(T) - f(0, T)$ and $1_{S(T) > K}$.

}

\section{Calibration}
Before using a model to price any derivatives, we usually calibrate it on the vanilla market which
means that it is able to price vanilla options with the concerned model and the resulting implied
volatilities match the market-quoted ones.  More precisely, it is necessary to determine all parameters
presenting in the different stochastic processes which define the model.  In such a way, all the European
option prices derived in the model are as consistent as possible with the corresponding market ones.\\

The calibration procedure for the two-factor model with local volatility can be decomposed into
three steps: 

\begin{itemize}
\item Parameters present in the one-factor dynamic for the interest rates, $\mu(t, r(t))$ and $\alpha(t,r(t))$, are chosen to match European swaption / cap-foors values. Methods for doing so are well developed in the literature (see
e.g \cite{BrigoMercurio06}). 

\item The correlation parameter $\rho$ is typically chosen either by historical estimation
or from occasionally observed prices of hybrid product involving interest rate and the underlying spot (see discussion in \cite{GobetHok14}).

\item After these two steps, the calibration problem
consists in finding the local volatility function $\sigma(t, S(t))$  which is consistent with its associated implied volatility surface.

\end{itemize}

Here we focus on the third step of the calibration and propose to use some martingales properties of the model to allow an efficient implementation.\\

Let's introduce the projection of the discount factor on the state variables $(S(t), r(t))$, defined as 
\begin{equation} \label{zetafunction}
\mathcal{Z}(t) := E[Z(t) | S(t), r(t)] = \mathcal{Z}(t, S(t), r(t)),
\end{equation}
with $Z(t)$ given in (\ref{discountingfactor}) and assume that $\mathcal{Z} \in \cC^{1, 2}$ on $[0, T] \times \R^2$. \\

Our objective is to determine the product functional $P(t, S, r) \cZ(t,S,r)$ where 
$P(t, S, r)$ represents the joint distribution of $(S(t), r(t))$. Indeed, we write

\begin{align} 
E[e^{-\int_0^T r(s)ds} (r(T) - f(0, T))1_{S(T) > K}] &= E[E[e^{-\int_0^T r(s)ds} (r(T) - f(0, T))1_{S(T) > K} / S(t), r(t)]] \\
 &= E[ \cZ(T, S(t), r(t)) (r(T) - f(0, T))1_{S(T) > K}] \\ 
 &= \int \left( r - f(0, T) \right) 1_{ S \geq K} (P\mathcal{Z})(T, S, r) dSdr.
 \label{Extra_term_evaluation}
\end{align}

We can then compute the corrective term (\ref{LVStoIR}) at least numerically.
As $(S(t), r(t))$ are the model state variables, $P(t, S, r) \cZ(t,S,r)$ can also be used for option pricing.\\

For any fix $T > 0$ and $h(S,r)$ a Borel-measurable function, let's define the function

\begin{equation}\label{fdefinition}
f(t, S, r) = E^{t, S, r}[e^{-\int_t^T r(s)ds}h(S(T), r(T))],
\end{equation}

\vspace{0.3cm}

where we assume  $E^{t, S, r}|h(S(T), r(T))| < +\infty$ for all $t, S, r$.\\

Using martingale argument as for the discounted Feynman-Kac theorem (see theorem 6.4.3 in \cite{Shreve04}), we derive
the partial differential equation satisfied by  $f(t, S, r)$:
\begin{align}\label{PDEf}
f_t + rSf_s + \mu f_r + \frac{1}{2}S^2 \sigma^2 f_{ss} + \frac{1}{2} \alpha^2 f_{rr} 
+ \rho \sigma \alpha S f_{sr} - rf = 0,
\end{align}
with the terminal condition

\begin{equation}\label{TerminalCondition}
f(T, S, r) = h(S, r) \: \: \forall (S, r).
\end{equation}

In the risk neutral pricing framework, we have the following result

\begin{proposition}
For $\mathcal{Z}(t)$ and $f$ defined, respectively, as in (\ref{zetafunction}) and in (\ref{fdefinition}), $\mathcal{Z}(t) f(t, S(t), r(t))$ is a martingale.
\end{proposition}

\begin{proof}

Using the martingality of $Z(t)f(t, S(t), r(t))$ and the Markov property of solutions in (\ref{OriginalSDEs}), for $t \geq s$, we write

\begin{align}
E[Z(t)f(t, S(t), r(t))  / \cF_s ] & =  Z_sf(s, S_s, r_s),  \\
E[ E[Z(t)f(t, S(t), r(t)) / \cF_t] /  \cF_s ] & =  E[ Z_sf(s, S_s, r_s)   / \cF_s ],\\
E[ E[Z(t)f(t, S(t), r(t)) / S(t), r(t)] / \cF_s ] & =   E[ Z_sf(s, S_s, r_s)   /  S_s, r_s],\\
E[ \mathcal{Z}(t) f(t, S(t), r(t))/ \cF_s ] & =  \mathcal{Z}_s f(s, S_s, r_s).
\end{align}

\end{proof}

Since the martingale $\mathcal{Z}(t, S(t), r(t)) f(t, S(t), r(t))$ is an It\^o process, it must have zero drift. 
Calculating the drift term using It\^o formula and setting it to zero give the following corollary:

\begin{corollary}  $\mathcal{Z}(t,S,r) f(t,S,r)$ is a solution of the following partial differential equation 
\begin{align}\label{PDEZf}
( \cZ f)_t + rS (\cZ f)_s + \mu (\cZ f)_r + \frac{1}{2}S^2 \sigma^2 (\cZ f)_{ss} + \frac{1}{2} \alpha^2 (\cZ f)_{rr} 
+ \rho \sigma \alpha S(\cZ f)_{sr} = 0.
\end{align}
\end{corollary}

Next proposition provides the forward equation satisfied by $P(t, S, r) \cZ (t, S, r)$.

\begin{proposition} \label{FPDEforPZ}
Assume the probability density $P(t, S, r)$ and its derivatives decay fast enough to $0$ for large 
$|(S,r)|$  to preclude boundary terms, i.e., for $|(S,r)| \to \infty$ 

\begin{equation}\label{fastdecayboundaryterms}
\begin{array}{ccc}
\mu P\cZ  = \sigma^2 s^2 P\cZ = \alpha^2 P\cZ =  s\sigma \alpha P\cZ &= 0 ,\\
 \frac{\partial(\sigma^2 s^2 P\cZ)}{\partial S} = \frac{\partial(\alpha^2P\cZ)}{\partial r} = 
 \frac{\partial(s \sigma \alpha P\cZ)}{\partial r} = \frac{\partial(s \sigma \alpha P\cZ)}{\partial s} &= 0.
\end{array}
\end{equation}
Then $P\cZ$ satisfies the following forward equation with the Dirac delta function as the initial condition

\begin{equation} \label{fwdeqpz}
\left\lbrace
\begin{array}{l}
 (P \cZ )_t + (rS P\cZ )_s + (\mu P\cZ )_r  -\frac{1}{2}(s^2 \sigma^2 P\cZ )_{ss}  
-\frac{1}{2}(\alpha^2 P\cZ )_{rr} - \rho (\alpha \sigma SP \cZ)_{sr} + r (P \cZ ) = 0,\\
(P \cZ )(0, S, r) = \delta(S-S_0, r-r_0).
	\end{array}
\right.
\end{equation}
or
\begin{equation} \label{fwdeqpzexpansion}
\left\lbrace
\begin{array}{l}
 (P \cZ )_t + \left[ S(r - 2\sigma^2 - \rho\sigma  \alpha_r ) -2\sigma\sigma_S S^2  \right](P\cZ )_s + 
 \left[ \mu -2 \alpha \alpha_r - \alpha\rho(\sigma +S \sigma_s) \right](P\cZ )_r  \\
  -\frac{1}{2}S^2 \sigma^2 (P\cZ)_{ss} -\frac{1}{2}\alpha^2 (P\cZ )_{rr} - \rho \alpha S\sigma (P \cZ)_{sr} \\
  + \left[ 2r + \mu_r - (\sigma^2 + 4S \sigma \sigma_s + S^2(\sigma_{ss} \sigma + \sigma_s^2)) 
  - (\alpha \alpha_{rr} + \alpha_r^2) - \rho \alpha_r (\sigma + S \sigma_s) \right] (P \cZ ) = 0,\\
(P \cZ )(0, S, r) = \delta(S-S_0, r-r_0).
	\end{array}
\right.
\end{equation}

\end{proposition}

\begin{proof}

Let's note by
\begin{equation}
C_0^2(\R^2) = \{ h \in C^2(\R^2) \, \, with \, \,  compact \, \, support\}
\end{equation}
and consider $f(t, S, r)$ defined as in (\ref{fdefinition}) for $h \in C_0^2(\R^2)$.\\

Using the martingale property of $\cZ(t) f(t, S(t), r(t))$ , we write for any $t \geq 0$ 
\begin{align}
f(0, S(0), r(0)) & = E[ \cZ(t) f(t, S(t), r(t)) ]\\
			   & = \int [P(t, S, r)\cZ (t, S, r)] f(t, S, r) dSdr. \label{couplepz}
\end{align}
By taking derivative w.r.t $t$ in (\ref{couplepz}), using (\ref{PDEf}), performing integration by parts and using zero boundary conditions (\ref{fastdecayboundaryterms}), we get 
\begin{equation}
\int f [ (P \cZ )_t + (rS P\cZ )_s + (\mu P\cZ )_r  -\frac{1}{2}(s^2 \sigma^2 P\cZ )_{ss} 
-\frac{1}{2}(\alpha^2 P\cZ )_{rr} - \rho (\alpha \sigma s P \cZ)_{sr} + r (P \cZ ) ]dsdr = 0.
\end{equation} 
Note that when $t$ approaches $T$, the function $f(t, S, r)$ approaches $h(S,r)$ and the last equation becomes
\begin{equation}
\int h [ (P \cZ )_t + (rS P\cZ )_s + (\mu P\cZ )_r  -\frac{1}{2}(s^2 \sigma^2 P\cZ )_{ss} 
-\frac{1}{2}(\alpha^2 P\cZ )_{rr} - \rho (\alpha \sigma s P \cZ)_{sr} + r (P \cZ ) ]dsdr = 0.
\end{equation} 
Since $ h \in C_0^2(\R^2) $ is arbitrary, we obtain the forward equation (\ref{fwdeqpz}).\\

The Dirac delta initial condition means that at time $t = 0$ we are sure that the spot 
$S(0)$ equals $S_0$, the interest rate $r(0)$ equals $r_0$ and $\cZ(0) = 1$.
\end{proof}

\remark{We know (see e.g proposition 11.5 in \cite{AveLau00}) that $P(t, S, r)$ satisfies the forward Fokker-Planck equation given by

\begin{equation} \label{fwdeqp}
\left\lbrace
\begin{array}{l}
 (P)_t + (rS P )_s + (\mu P)_r  -\frac{1}{2}(s^2 \sigma^2 P)_{ss} 
-\frac{1}{2}(\alpha^2 P )_{rr} - \rho (\alpha \sigma SP)_{sr} = 0,\\
(P \cZ )(0, S, r) = \delta(S-S_0, r-r_0).
	\end{array}
\right.
\end{equation}

In comparison, the equation for $P \cZ $ in (\ref{fwdeqpz}) has an additional reaction term.
Also by definition, we have for all $t \geq 0$

\begin{align}
\int P(t, S, r)dSdr &=1, \\ 
\int (P\cZ)(t, S, r)dSdr &= E[ \mathcal{Z}(t, S(t), r(t)) ] = ZC(0, t). \label{IntegratePZ}
\end{align}

}

\section{Numerical methods}

The targeted equation to solve is (\ref{fwdeqpzexpansion})
with a Dirac delta function as initial condition. 
For the calibration, the space variables $(S, r)$ can potentially be defined in an  
unbounded domain typically $S \geq 0$ and $r \in \R$ (see section \ref{sec:Numerical experiments}).  This unbounded domain is actually truncated by $(S, r) \in [Smin, Smax] \times [rmin, rmax]$ in the finite difference spatial discretization where $Smin$ is close to 0, $|Smax|, |rmin|, |rmax|$ are sufficiently large in the numerical experiments (see \cite{WinForVet04}). \\

For the boundary values, we opt for Dirichlet type conditions.  
The values depend on the model considered. In our numerical experiments at section \ref{sec:Numerical experiments}, 
$r_t$ has Gaussian distribution for all $t$ and $\Q (S_t > 0, \, \forall \, t \geq 0) = 1$.
The values for $\cZ (t, S, r)$ are expected to be around 1 (see Figures \ref{figure:ZfunctionT1Y} and \ref{figure:ZfunctionT2Y}). 
Then, reasonably, the boundary values for $P\cZ(t, S, r)$ are set to be $0$. 
For the initial condition, we propose to approximate the Dirac delta function by the Gaussian kernels i.e a family of  Gaussian functions 

\begin{equation}
\gamma_{\frac{1}{N}}(S, r) =   \frac{1}{2 \pi \sqrt{ det \Sigma}} exp \left \{ -\frac{1}{2} \small{ \left(\begin{array}{c} S-S_0 \\ 
r - r_0\\
\end{array}\right)' } \Sigma^{-1}  \small{ \left(\begin{array}{c} S-S_0 \\ 
r - r_0\\
\end{array}\right) } \right \}
\end{equation}
parametrized by the parameter $N>0$ with $ \Sigma = \left(\begin{array}{cc} 
\frac{1}{N} & 0\\
0 & \frac{1}{N}
\end{array}
\right)$ (see e.g \cite{Zaid13} and \cite{SaiWo13}).\\

Constructing fully implicit scheme for solving equation (\ref{fwdeqpz}) may result some difficulties as the descritized linear system becomes banded structure and requires efficient linear solvers.  Besides if the employed grid numbers are large, the initialization of program can cause problem as the size of matrix is too large.  In the following, the details of doing time discretization of solving equation (\ref{fwdeqpz}) is explained by using Peaceman-Rachford scheme which can resolve these issues.

\subsection{ADI Solver}
Let $N_S,N_r$ be the number of grid points for doing spatial discretizations on the direction $S$ and $r$, respectively, and $N_t$ be the number of grid points employed on temporal discretization. The essential idea of Peaceman-Rachford scheme is to split the calculation in the time-marching scheme into several steps with respect to different spatial variables.  More precisely, the scheme used for evaluating $(P\cZ)^{n+1}$ from the known value $(P\cZ)^n$ can be specified as the following two steps.
\\
\\
\begin{large}
\textbf{ADI Step 1} \\
\end{large}
In the first step of ADI, the finite difference spatial discretization on $S$ is treated implicitly and the rest terms are explicitly for time step from $n$ to $n+\frac{1}{2}$, namely

\begin{align} \label{ADI_Step_1}
  \begin{split}
\frac{(P\cZ)_{ij}^{n+\frac{1}{2}} - (P\cZ)_{ij}^{n}}{\Delta t/2}
+ C_1\frac{(P\cZ)_{i+1,j}^{n+\frac{1}{2}} - (P\cZ)_{i-1,j}^{n+\frac{1}{2}}}{2\Delta S}
+ C_2\frac{(P\cZ)_{i,j+1}^{n} - (P\cZ)_{i,j-1}^{n}}{2\Delta r}\\
+ C_3\frac{(P\cZ)_{i+1,j}^{n+\frac{1}{2}} - 2(P\cZ)_{ij}^{n+\frac{1}{2}} + (P\cZ)_{i-1,j}^{n+\frac{1}{2}}}{(\Delta S)^2}
+ C_4\frac{(P\cZ)_{i,j+1}^{n} - 2(P\cZ)_{ij}^{n} + (P\cZ)_{i,j-1}^{n}}{(\Delta r)^2}\\
+ C_5\frac{(P\cZ)_{i+1,j+1}^{n}+(P\cZ)_{i-1,j-1}^{n}-(P\cZ)_{i-1,j+1}^{n}-(P\cZ)_{i+1,j-1}^{n}}{4\Delta S \Delta r} 
+ C_6 (P\cZ)_{ij}^{n+\frac{1}{2}} =0,
  \end{split}
\end{align}
where 
\[
C_1 = (r_jS_i-2S_i\sigma_i^2-2(S_i)^2\sigma_i(\sigma_S)_i-\rho\sigma_i S_i (\alpha_r)_j), \;
C_2 = (\mu_j - \rho\sigma_i\alpha_j - \rho(\sigma_S)_i S_i \alpha_j - 2 \alpha_j(\alpha_r)_j), 
\]
\[
C_3 = \left(\frac{-S_i^2\sigma_i^2}{2} \right), \; C_4 = \left(\frac{-\alpha_j^2}{2} \right),\; C_5 = \left(-\rho\sigma_i S_i \alpha_j \right),
\]
\[
C_6 = ( 2r_j+(\mu_r)_j-\sigma_i^2-4S_i\sigma_i(\sigma_S)_i -((\sigma_S)_i)^2S_i^2 - \sigma_i(\sigma_{SS})_iS_i^2 - ((\alpha_r)_j)^2 - \alpha_j(\alpha_{rr})_j-\rho(\sigma_S)_iS_i(\alpha_r)_j-\rho\sigma_i(\alpha_r)_j ),
\]
here $(P\cZ)_{ij}^{n}$ is the discretized notation of $(P\cZ)(n \Delta t,i\Delta S, j \Delta r)$ for $i = 1,...,N_S$, $j = 1,...,N_r$ and $n+\frac{1}{2}$ is a dummy time step.  Note the notations $(\sigma_S)_i, (\sigma_{SS})_i$ present that the evaluation of the sensitivities is at the discretized point $S_i$.  Similarly the notations $(\alpha_r)_j, (\alpha_{rr})_j$ show the values evaluated with the discretized point $r_j$.   Then it leads to solve 
\begin{equation} \label{Linear_system_ADI_Step_1}
H_1 (P\cZ)^{n+\frac{1}{2}} = f_1((P\cZ)^n).
\end{equation}
for each $j$, where $H_1$ is a tri-diagonal matrix with size $(N_S\times N_S)$ and entries equal to

\[
a_i = \frac{-C_1}{2\Delta S} + \frac{C_3}{(\Delta S)^2},\;
b_i = \frac{2}{\Delta t} - \frac{2C_3}{(\Delta S)^2} + C_6,\;
c_i = \frac{C_1}{2\Delta S} + \frac{C_3}{(\Delta S)^2},
\]
where $a_i, b_i, c_i$ are the entries for lower, main and upper diagonal for $i = 1,...,N_{S}$.  The right hand side value $f_1((P\cZ)^n)$ can be evaluated with the known value $(P\cZ)^n$ at current time step and is given by 

\begin{align}
f_1((P\cZ)^n) & = (P\cZ)_{ij}^{n} \left( \frac{2}{\Delta t} + \frac{2C_4}{(\Delta r)^2} \right)
		-(P\cZ)_{i,j+1}^{n} \left( \frac{C_2}{2\Delta r} + \frac{C_4}{(\Delta r)^2} \right)
		+ (P\cZ)_{i,j-1}^{n} \left( \frac{C_2}{2\Delta r} - \frac{C_4}{(\Delta r)^2} \right) \nonumber \\
		& -C_5\frac{ \left( (P\cZ)_{i+1,j+1}^{n}+(P\cZ)_{i-1,j-1}^{n}-(P\cZ)_{i-1,j+1}^{n}-(P\cZ)_{i+1,j-1}^{n} \right)}{4\Delta S \Delta r}.
\end{align}\\
\begin{large}
\textbf{ADI Step 2} \\
\end{large}
Once the solution at the dummy time step $(P\cZ)^{n+\frac{1}{2}}$ is obtained, the second step of ADI is to consider the finite difference spatial discretization on $r$ which is treated implicitly and the rest terms are explicitly for time step from $n+\frac{1}{2}$ to $n+1$ as follows
\begin{align} \label{ADI_Step_2}
  \begin{split}
\frac{(P\cZ)_{ij}^{n+1} - (P\cZ)_{ij}^{n+\frac{1}{2}}}{\Delta t/2}
+ C_1\frac{(P\cZ)_{i+1,j}^{n+\frac{1}{2}} - (P\cZ)_{i-1,j}^{n+\frac{1}{2}}}{2\Delta S}
+ C_2\frac{(P\cZ)_{i,j+1}^{n+1} - (P\cZ)_{i,j-1}^{n+1}}{2\Delta r} \\
+ C_3\frac{(P\cZ)_{i+1,j}^{n+\frac{1}{2}} - 2(P\cZ)_{ij}^{n+\frac{1}{2}} + (P\cZ)_{i-1,j}^{n+\frac{1}{2}}}{(\Delta S)^2}
+ C_{4}\frac{(P\cZ)_{i,j+1}^{n+1} - 2(P\cZ)_{ij}^{n+1} + (P\cZ)_{i,j-1}^{n+1}}{(\Delta r)^2}\\
+ C_{5}\frac{  (P\cZ)_{i+1,j+1}^{n+\frac{1}{2}}+(P\cZ)_{i-1,j-1}^{n+\frac{1}{2}}-
(P\cZ)_{i-1,j+1}^{n+\frac{1}{2}}-(P\cZ)_{i+1,j-1}^{n+\frac{1}{2}} }{4\Delta S \Delta r} 
+ C_{6} (P\cZ)_{ij}^{n+1} =0,
  \end{split}
\end{align}
for $i = 1,...,N_S$, $j = 1,...,N_r$ and $n+1$ is the time step of the pursued solution. Then it leads to solve 
\begin{equation} \label{Linear_system_ADI_Step_2}
H_2 (P\cZ)^{n+1} = f_2((P\cZ)^{n+\frac{1}{2}}),
\end{equation}
for each $i$, where $H_2$ is a tri-diagonal matrix with size $(N_r\times N_r)$ and entries equal to
\[
d_j = \frac{-C_2}{2\Delta r} + \frac{C_{4}}{(\Delta r)^2},\;
e_j = \frac{2}{\Delta t} - \frac{2C_{4}}{(\Delta r)^2} + C_{6},\;
f_j = \frac{C_2}{2\Delta r} + \frac{C_{4}}{(\Delta r)^2}.
\]
where $d_j, e_j, f_j$ are the entries for lower, main and upper diagonal for  $j = 1,...,N_{r}$. The right hand side value $f_2((P\cZ)^{n+\frac{1}{2}})$ again can be evaluated with the calculated solution $(P\cZ)^{n+\frac{1}{2}}$ at the dummy time step and is given by 

\begin{align}
f_2((P\cZ)^{n+ \frac{1}{2}}) & = (P\cZ)_{ij}^{n + \frac{1}{2}} \left( \frac{2}{\Delta t} + \frac{2C_3}{(\Delta S)^2} \right)
		-(P\cZ)_{i+1,j}^{n+ \frac{1}{2}} \left( \frac{C_1}{2\Delta S} + \frac{C_3}{(\Delta S)^2} \right)
		+ (P\cZ)_{i-1,j}^{n+ \frac{1}{2}} \left( \frac{C_1}{2\Delta S} - \frac{C_3}{(\Delta S)^2} \right) \nonumber \\
		& -C_5\frac{ \left( (P\cZ)_{i+1,j+1}^{n+ \frac{1}{2}}+(P\cZ)_{i-1,j-1}^{n+ \frac{1}{2}}-(P\cZ)_{i-1,j+1}^{n+ \frac{1}{2}}-(P\cZ)_{i+1,j-1}^{n+ \frac{1}{2}} \right)}{4\Delta S \Delta r}.
\end{align}

\remark{
\begin{itemize}
\item In the ADI steps 1 and 2, the coefficients $C_i, \, i = 1, .., 6$ are evaluated at time $t_n$. It corresponds to some approximations and allows simplifications.

\item Certain PDE schemes for solving the forward Fokker-Planck equation for 
$P(t, S, r)$ lead to some negative probabilities and are source of instability (see e.g \cite{Itkin17,sepp10}). In our numerical examples in section \ref{sec:Numerical experiments}, we have 
not seen these issues when solving the equation (\ref{fwdeqpzexpansion}) for $P\cZ$.

\item Using the concept of generator and \textit{M}-matrix in \cite{Itkin17}, Itkin proposed PDE spliting scheme s.t the integral of the discrete solution $P(t, S, r)$ is equal to $1$. Here, our solution for $(P\cZ)(t, S, r)$ should theoretically satisfy the relation (\ref{IntegratePZ}).
In our experiments in section \ref{sec:Numerical experiments}, integrating the numerical solution can give slightly different results. In that case, we normalize the numerical solution for $(P\cZ)(t, S, r)$ s.t relation (\ref{IntegratePZ}) is satisfied. 
 
\end{itemize}
}

\begin{algorithm}  
  \SetAlgoLined
  \caption{Alternative Direction Implicit Scheme}
  \KwIn{initial condition $(P\cZ)^0 = \delta(S-S_0, r-r_0)$, parameters}
  \KwOut{$(P\cZ)^{N_t}$}
  \For{$n = 0 : N_t-1$}
  {
    1. ADI step 1: \\
    \For{$j = 1 : N_r-1$}
    {
      Solve $H_1 (P\cZ)^{n+\frac{1}{2}} = f_1((P\cZ)^n)$ to get $(P\cZ)^{n+\frac{1}{2}}$ \;
    }    
    2. ADI step 2: \\
    \For{$i = 1 : N_S-1$}
    {
      Solve $H_2 (P\cZ)^{n+1} = f_2((P\cZ)^{n+\frac{1}{2}})$ to get $(P\cZ)^{n+1}$ \;
    } 
  }
  \label{algo:PDE_ADI}
\end{algorithm}

\subsection{Efficient implementation of the corrective terms} \label{EfficientImplementation}

Let's note by $K_1 < K_2< ...., < K_N$ the strike grid where we want to compute the corrective terms 
$E[Z(T) (r(T) - f(0, T))1_{S(T) > K}]$ at each maturity date $T$. We observe the corrective terms for two consecutive strikes are highly correlated. Indeed for $K_i < K_{i+1}$, we have

\begin{equation}\label{RelationCorrectiveTerms}
Adj(K_i) =  Adj(K_{i+1}) + E[Z(T) (r(T) - f(0, T))1_{K_{i+1} \geq S(T) > K_i}],
\end{equation} 
with $Adj(K) = E[Z(T) (r(T) - f(0, T))1_{S(T) > K}]$.\\

So an efficient algorithm consists to calculate $Adj(K_N)$ first and then to use the relation 
(\ref{RelationCorrectiveTerms}) to compute sequentially the other terms. With this method, computing all corrective terms for a given maturity consists essentially to perform one numerical integration in the whole domain of discretization for $S$ and $r$, which allows to speed up significantly the computation time. 
In our numerical experiments, the integrals for calibration and pricing are estimated numerically using the PDE solutions and grid points. 

\begin{algorithm}  
  \SetAlgoLined
  \caption{Calibration Algorithm}
  \KwIn{Necessary parameters for using the ADI method.}
  \KwOut{$\sigma(T_i,K_j)$, $i=1,...,N_T$, $j=1,...,N_K$.}
  \For{$i = 1 : N_T$}
  {
	Solve equation (\ref{fwdeqpz}) to get $(P\cZ)$   \; 
    \For{$j = 1 : N_K$}
    {
      1. Do numerical integration using equation (\ref{RelationCorrectiveTerms}) to obtain the extra term;\\
      2. Calculate the sensitivities to get $\sigma_{Dup}$\;
      3. Evaluate $\sigma(T_i,K_j)$ \;
    }
  }
  \label{algo:Calibration}
\end{algorithm}

\section{Numerical experiments} \label{sec:Numerical experiments}
\subsection{ Black-Scholes Hull-White hybrid model}

Let's consider the Black-Scholes economy with Hull-White stochastic interest rates model:

\begin{equation} \label{BSVasicekSDEs}
\left\lbrace
\begin{array}{l}
	\frac{dS(t)}{S_t} = r(t)dt + \sigma_1 dW^1(t) , \hspace{0.3cm} S(0) = S_0 ,\\
	dr(t) = a(\theta(t)-r(t))dt +  \sigma_2 (\rho dW^1(t) + \sqrt{1-\rho^2} dW^2(t)), \hspace{0.3cm} r(0) = r_0. \\
	\end{array}
\right.
\end{equation}

$a$ represents the speed of reversion, $\theta$ the long term mean level.
From \cite{BrigoMercurio06}, the zero coupon $ZC(0,T)$ price and the forward rates are given respectively by

\begin{equation}\label{ZCVasicek}
ZC(0,T) = A(0, T)e^{-B(0, T)r(0)},
\end{equation}

\begin{equation}\label{FwdVasicek}
f(0,T) = -\frac{\sigma_2^2}{2a^2} + \theta - (\theta- \frac{\sigma_2^2}{a^2} - r0)e^{-aT} - \frac{\sigma_2^2}{2a^2}e^{-2aT},
\end{equation}
with $B(0, T) = \frac{1}{a}(1-e^{-aT})$ and $A(0, T) = e^{ \left[ (\theta - \frac{\sigma_2^2}{2a^2})(B(0, T)-T) 
- \frac{\sigma_2^2}{4a} B(0, T)^2\right] }$. \\

\remark{ We can exactly fit the term structure of interest rates being observed in the market  by considering $\theta = \theta(t)$ . From \cite{BrigoMercurio06}, the formula is 
 given by

\begin{equation}
\theta(t) = \frac{1}{a}\frac{\partial f(0,t)}{\partial t} + f(0,t) + \frac{1}{2}\left( \frac{\sigma_2}{a} \right)^2(1-e^{-2at}).
\end{equation}

}

Using the martingale method, the European call option price $C(T,K)$ with maturity $T$ and strike $K$ is derived in \cite{Fang12} which we summarize in the next proposition.

{
\proposition{\label{BlackScholesStochasticRates}
\begin{equation}
C(T, K) = S_0N(d_1) - KZC(0,T)N(d_2),
\end{equation}

where $n(t) = \frac{1}{\sqrt{2 \pi}} e^{-\frac{1}{2}t^2}$, 
$N(x) = \int_{-\infty}^{x}n(t)dt$, 
$k = ln(K)$,

$d_1 = \frac{ \log \left( \frac{S_0}{K} \right) - \log(ZC(0,T)) + \frac{1}{2} \int_0^T \hat{\sigma}^2(t)dt}{\sqrt{\int_0^T \hat{\sigma}^2(t)dt}}$, \, \, $d_2 = d_1 - \sqrt{\int_0^T \hat{\sigma}^2(t)dt}$, \\
$\hat{\sigma}(t) = \sqrt{ \sigma_1^2 + 2\rho\sigma_1 \sigma_2 X(t) + \sigma^2_2 X^2(t)} $ and 
$X(t) = -\frac{1}{ZC(0,T)} \frac{\partial ZC(0,T)}{\partial r}$.\\
}
}

From (\ref{ZCVasicek}), we have $X(t) = B(0, t)$ and

\begin{align}
 \int_0^T \hat{\sigma}^2(t)dt& =  \sigma_1^2T + \frac{2 \rho \sigma_1 \sigma_2}{a} \left[T + \frac{e^{-aT}-1}{a} \right]
  + \frac{\sigma_2^2}{a^2} \left[ T - \frac{1}{2a}(3 - 4e^{-aT} + e^{-2aT}) \right].
\end{align}

The following corollary provides analytical formulas for $C_T, \, C_K$ and $C_{KK}$
used in the local volatility calibration expression (\ref{LVStoIR}).

{
\corollary{ \label{BSStochasticIRGreeks} Using the definitions in proposition (\ref{BlackScholesStochasticRates}) and $g(T) = \int_0^T \hat{\sigma}^2(t)dt$, we have

\begin{align}
\frac{\partial C}{\partial T}(T, K)&:= C_T(T,K) = \frac{S_0 n(d_1)}{2} \frac{\hat{\sigma}^2(T)}{\sqrt{g(T)}} + K ZC(0, T)f(0, T) N(d_2), \label{GreekCT} \\
\frac{\partial C}{\partial K}(T, K) &:= C_K(T,K) = - ZC(0, T)f(0, T) N(d_2), \label{GreekCK}\\
\frac{\partial^2 C}{\partial K^2}(T, K) &:= C_{KK}(T,K) = \frac{ZC(0, T) n(d_2)}{K \sqrt{g(T)}}. \label{GreekCKK}
\end{align}

}
}

The proof is given in appendix.\\

This model offers tractability and we can compute the function  $\cZ (t,y,r)$ analytically whose expression is given in the next proposition.

\begin{proposition}\label{ZtinBSVasicekModel}
Let's define $Y(T) := log(S(T))$, $R(T) := \int_0^T r(s)ds$ and assume the matrix $\Sigma_{yr}$, defined in (\ref{Matyr}), invertible. Then we have
\begin{equation} \label{ZTsrInVasicekBSModel}
E[e^{-R(T)} | Y(T), r(T)] = exp \left\{ {-\mu_R - \Sigma_{yr, R}^{t}\Sigma^{-1}_{yr}}. {\small{ \left( \begin{array}{c} Y(T)-\mu_y \\ 
r(T) - \mu_r\\
\end{array}   \right) } } +  \frac{\Sigma_R - \Sigma^{t}_{yr, R} \Sigma^{-1}_{yr} \Sigma_{yr, R}  }{2}\right\},
\end{equation}
with 
\begin{align}
\mu_y & =  log(S(0)) + \mu_R - \frac{1}{2}\sigma_1^2 T,  \\
\mu_r & = r(0)e^{-aT} + \theta(1-e^{-aT}),\\
\mu_R & = r(0) \frac{(1-e^{-aT})}{a} +   \theta T - \frac{\theta}{a}(1-e^{-aT}), \\
\Sigma_y & = \Sigma_R + \sigma_1^2 T + \frac{2 \sigma_1 \sigma_2 \rho}{a} \left[ T - \frac{1}{a}(1-e^{-aT}) \right],\\
\Sigma_r & = \frac{\sigma_2^2 }{2a}(1-e^{-2aT}),\\
\Sigma_R & = \left(\frac{\sigma_2}{a}\right)^2 \left[ T + \frac{1}{2a}(1-e^{-2aT}) - \frac{2}{a}(1-e^{-aT}) \right],\\
\Sigma_{yr} & = \left(
				\begin{array}{cc}
				\Sigma_y & \frac{\rho \sigma_1 \sigma_2 }{a}(1-e^{-aT}) \\
				\frac{\rho \sigma_1 \sigma_2 }{a}(1-e^{-aT}) &  \Sigma_r
			    \end{array}
			  \right),		\label{Matyr}\\
\Sigma_{yr,R} & = \left(
				\begin{array}{c}
				 \frac{\rho \sigma_1 \sigma_2 }{a} [T - \frac{1}{a}(1-e^{-aT})]\\
				 \left(\frac{\sigma_2}{a}\right)^2( \frac{1}{2} - e^{-aT} + \frac{1}{2}e^{-2aT} )
			    \end{array}
			  \right).		  
\end{align}

\end{proposition}

\begin{proof}

Standard computations give

\begin{align}
Y(T) & = \log(S(0)) + R(T) - \frac{1}{2}\sigma_1^2 T  + \sigma_1 W^1(T), \\
r(T) & = r(0)e^{-aT} + \theta (1-e^{-aT}) + \sigma_2 \int_0^T e^{-a(T-t)} \left( \rho dW^1(t) + \sqrt{1-\rho^2} dW^2(t) \right), \\
R(T) & = \frac{1}{a}(1-e^{-aT}) (r(0)-\theta) + \theta T  + \frac{\sigma_2}{a} \int_0^T (1-e^{-a(T-t)}) \left( \rho dW^1(t) + \sqrt{1-\rho^2} dW^2(t) \right). 
\end{align}

Then $\left(\begin{array}{c} 
				 Y(T)\\
				 r(T)\\
				 R(T)
			    \end{array}
			  \right)$ is a Gaussian vector with mean 
			  $ \mu = \left(\begin{array}{c} 
				 \mu_y\\
				 \mu_r\\
				 \mu_R
			    \end{array}
			  \right)$ and covariance matrix  $\left(
				\begin{array}{cc}
				 \Sigma_{yr}  &  \Sigma_{yr,R}\\
				 \Sigma_{yr,R}^t& \Sigma_R
			    \end{array}
			  \right)$. It is well-known (see e.g chap 2.3 in \cite{Glasserman03}) the conditional distribution of $R(T)$  given by $\left(\begin{array}{c} 
				 Y(T)\\
				 r(T)				
			    \end{array}
			  \right)$ is a normal random variable with mean $\mu_{R|y,r}$ and variance $\sigma^2_{R|y,r}$ given respectively by
			
			\begin{align}
\mu_{R|y,r} & =  \mu_R + \Sigma_{yr,R}^t \Sigma_{yr}^{-1}. \left(\begin{array}{c} 
				 Y(T) - \mu_y\\
				 r(T) - \mu_r				 
			    \end{array}
			  \right),\\
\sigma^2_{R|y,r} & = \Sigma_R - \Sigma_{yr,R}^t \Sigma_{yr}^{-1} \Sigma_{yr,R}.
\end{align} Using the moment generating function for a standard normal random variable $N$, 
 \begin{equation}
  E[e^{uN}] = e^{\frac{u^2}{2}},\: \forall u \in \R,
 \end{equation}
then we obtain the expression in equation (\ref{ZTsrInVasicekBSModel}).  
			  
\end{proof}

\subsubsection{Tests}

For the numerical tests, we consider two sets of parameters:

\begin{center}
\begin{tabular}{|c|c|c|}
\hline
parameters & set 1 & set 2\\
\hline
$S_0$ 	& 1 & 1\\
\hline
$r_0$ 	& $2 \%$ & $2 \%$\\
\hline
$\sigma_1$ 	& $20 \%$ & $20 \%$\\
\hline
$\sigma_2$ 	& $4 \%$ & $4 \%$\\
\hline
$\rho$ 	& $40 \%$ & $-40 \%$\\
\hline
$a$ 	& $0.5$ & $0.5$\\
\hline
$\theta$ 	& $2 \%$ & $2 \%$\\
\hline
$T$ 	& $ 1 $ & $ 2$\\
\hline
\end{tabular}
\end{center}

These model parameters correspond to the order of magnitude usually used in the equity and interest rates derivatives pricing. Also they are in line with the statistical estimates of model parameters performed by Kim in \cite{Y.JK02}.
Here we consider higher volatility $\sigma_2$ for the dynamic of interest rates and correlations values $\rho$ to better measure the impact of stochastic rates process.\\

For each set of model parameters, we solve the PDE (\ref{fwdeqpzexpansion}) in an uniform grid  using the ADI method explained in the previous section. For set 1 and set 2 respectively, the sizes of doing temporal and spatial discretizations are given by 
$ds = 0.0156, \, dr = 0.0026, \, dt = 0.0099$ and $ds = 0.025, \, dr = 0.0037, \,dt = 0.019$.  \\

For set 1, the analytic solutions for $P \cZ$ are shown in Figure \ref{figure:PZFormulaT1Y}, numerical solutions in Figure \ref{figure:PZNumericalT1Y} and their discrepancy in Figure \ref{figure:PZDiscrepancyT1Y}.
For set 2, similar results are illustrated in Figure \ref{figure:PZFormulaT2Y}, Figure \ref{figure:PZNumericalT2Y} and 
Figure \ref{figure:PZDiscrepancyT2Y}. For both sets, we observe the discrepancies between analytic formula and numerical solutions on the PDE grid are reasonable {\it{small}}. For a more accurate assessment, we perform option pricing with the numerical solution  $P \cZ$ at maturity $T$ for each set. We evaluate numerically the European call option prices for various strikes. Numerical and analytical prices solutions are illustrated respectively in Figure 
\ref{figure:CallPricesT1Y} and Figure \ref{figure:CallPricesT2Y} for set 1 and 2. The discrepancies are shown respectively in Figure \ref{figure:PricesDiscrepancyT1Y} and Figure
\ref{figure:PricesDiscrepancyT2Y}.
For both cases, we observe very good pricing accuracy as the differences are within couple of basis points 
($10^{-4}$).\\

Also we quantify the impact of the corrective term in expression (\ref{LVStoIR}) for stochastic interest rates.
To avoid any scaling, we show the analytic formula in the numerator i.e $E[Z(T) (r(T) - f(0, T))1_{S(T) > K}]$ on the grid $(T, K)$. The results are illustrated in Figure \ref{figure:CorrectiveTermNumeratorSet1} and Figure \ref{figure:CorrectiveTermNumeratorSet2} respectively for set 1 and 2.
In the first case with positive correlation parameter $\rho = 0.4$, we observe positive corrective terms where higher values are concentrated around the moneyness $1$. It is expected as discussed in the remark of section  
2 where the corrective term measures the covariance between interest rates and equity spot.
Similarly, for set 2 where the correlation parameter is negative $\rho = -0.4$, we obtain negative values with higher levels concentrated around the moneyness $1$.\\

Finally, we draw the projected discounted factor $\mathcal{Z}(T, S, r)$ in Figure \ref{figure:ZfunctionT1Y} and Figure \ref{figure:ZfunctionT2Y} for set 1 and 2 respectively. In both cases, the forms are similar and their levels vary around $1$.

\begin{figure}[htbp]
  \centering
  \includegraphics[width=0.8\textwidth]{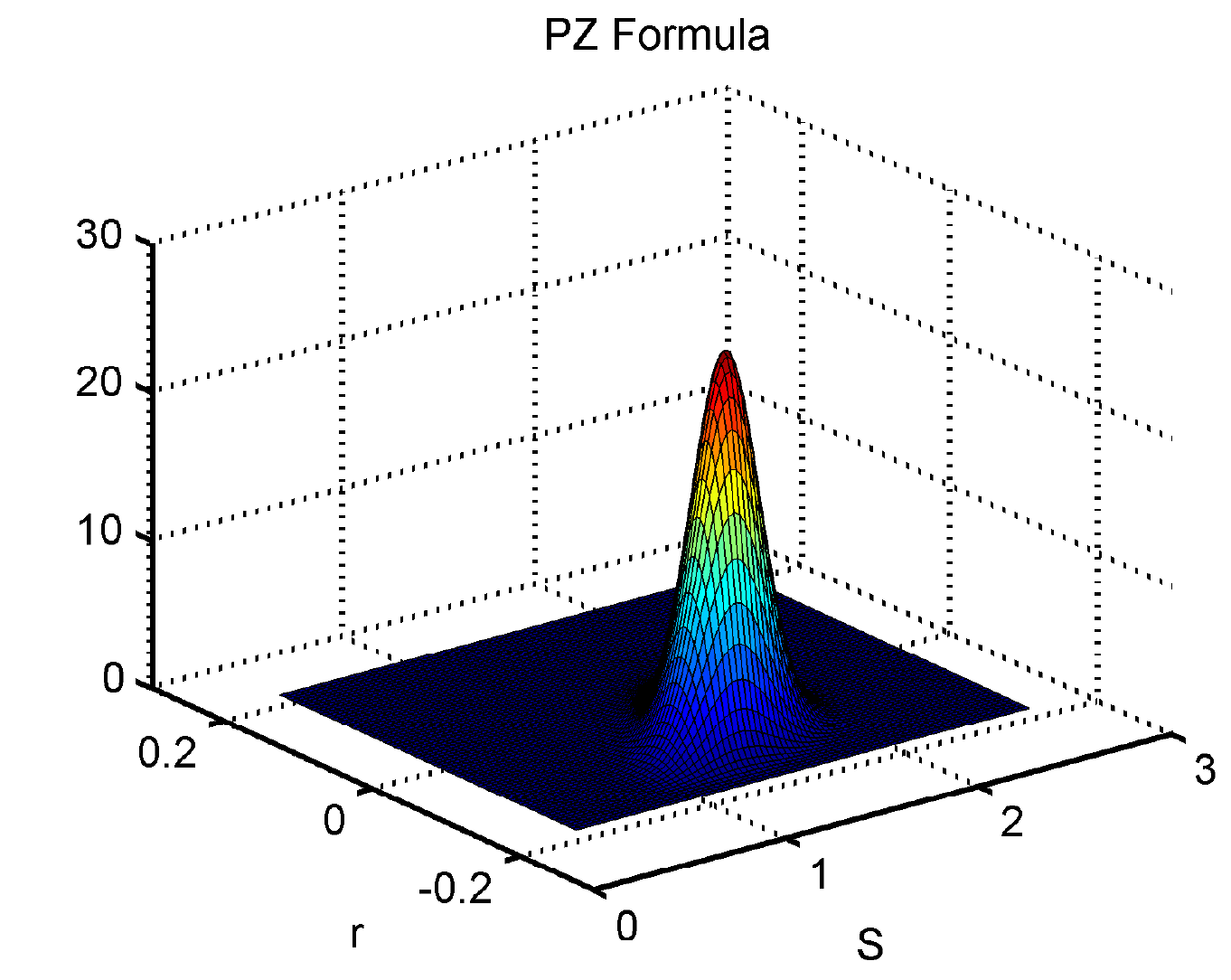}
  \caption{ Parameters $S_0 = 1.0$, 
$r_0 = 0.02$, $\sigma_1 = 0.2$, $\sigma_2 = 0.04$, $\rho = 0.4$, $a = 0.5$,
$\theta = 0.02$, $T = 1.0$.}
  \label{figure:PZFormulaT1Y}
\end{figure}

\begin{figure}[htbp]
  \centering
 \includegraphics[width=0.8\textwidth]{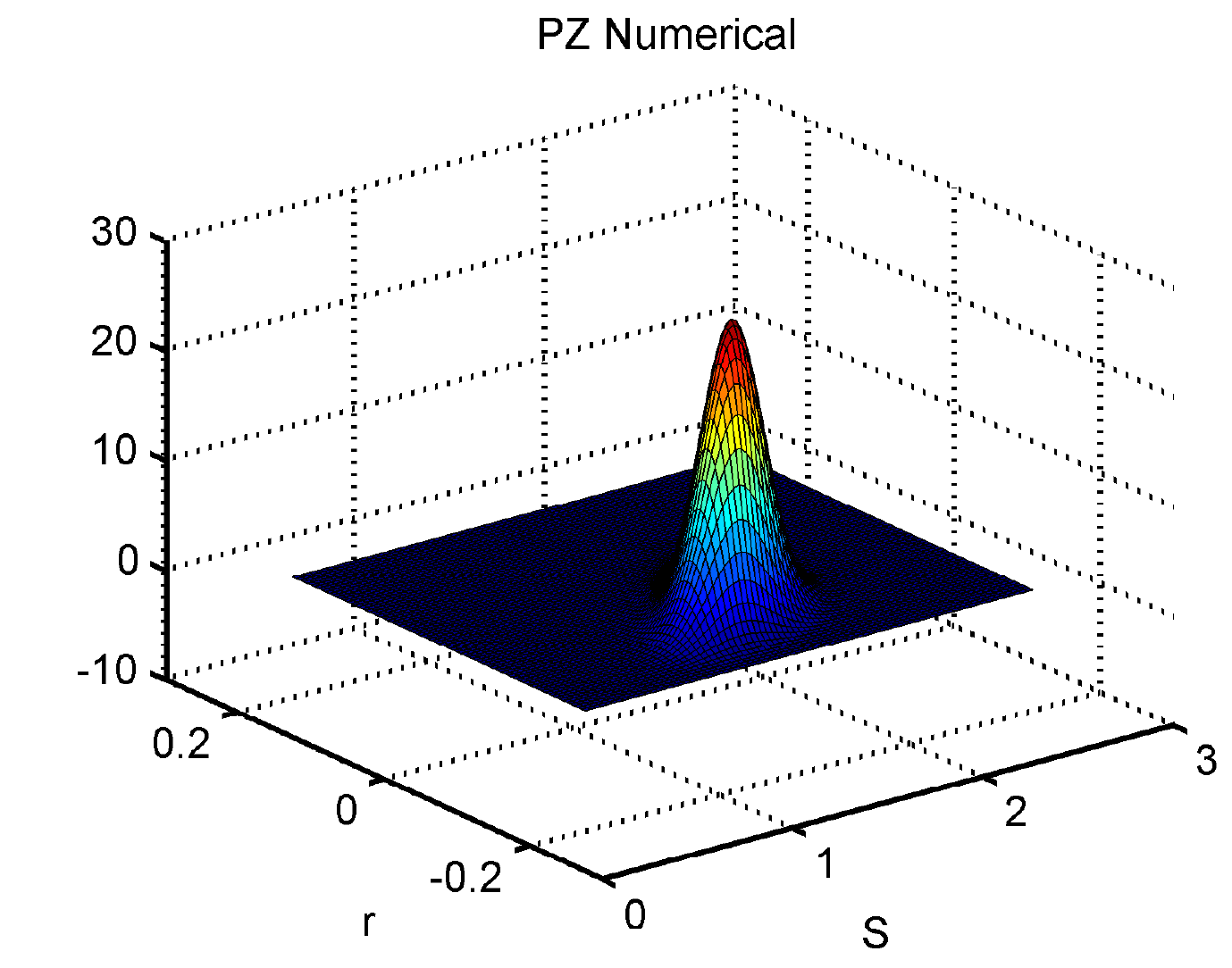}
  \caption{ Parameters $S_0 = 1.0$, 
$r_0 = 0.02$, $\sigma_1 = 0.2$, $\sigma_2 = 0.04$, $\rho = 0.4$, $a = 0.5$,
$\theta = 0.02$, $T = 1.0$.}
  \label{figure:PZNumericalT1Y}
\end{figure}

\begin{figure}[htbp]
  \centering
 \includegraphics[width=0.8\textwidth]{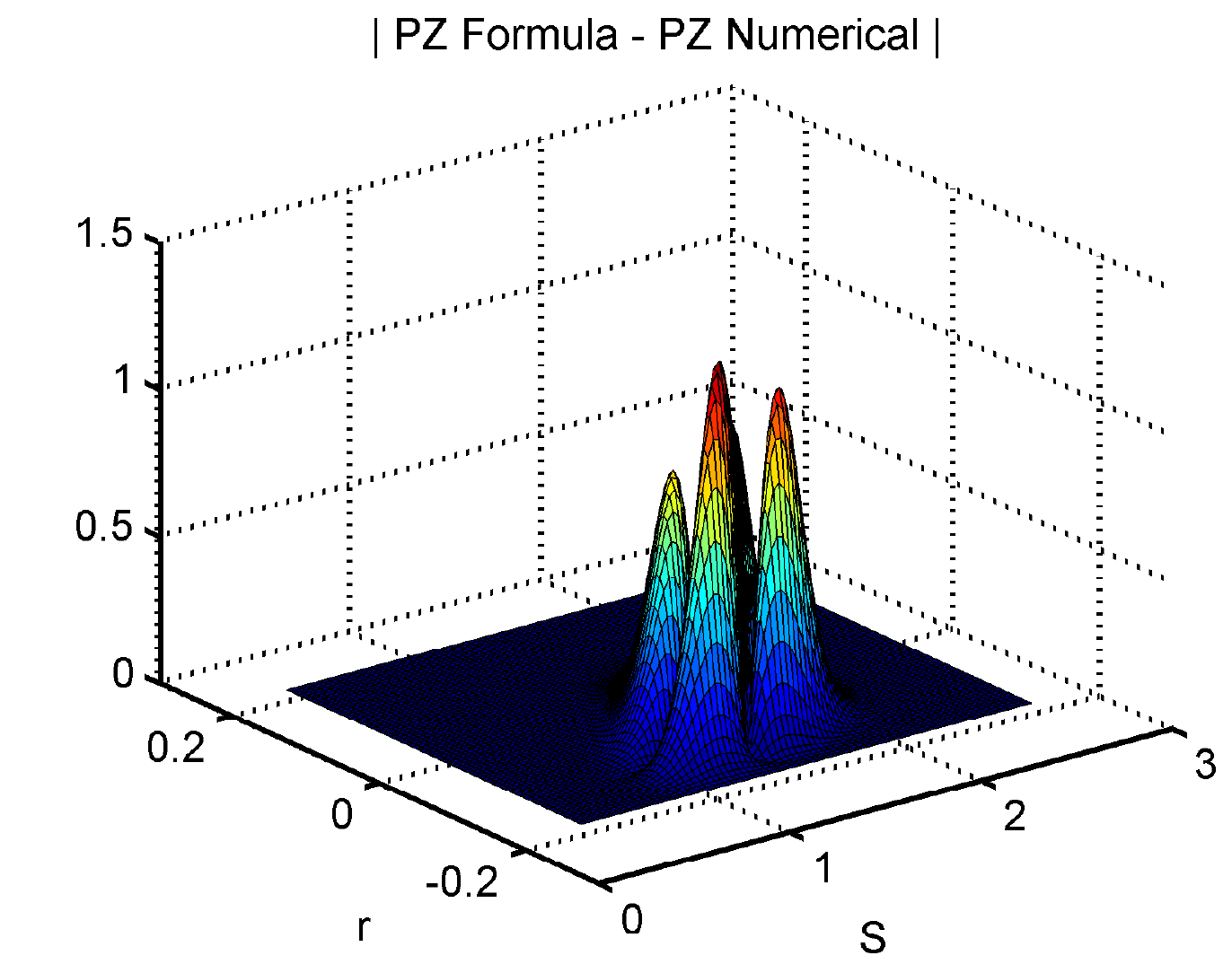}
  \caption{ Parameters $S_0 = 1.0$, 
$r_0 = 0.02$, $\sigma_1 = 0.2$, $\sigma_2 = 0.04$, $\rho = 0.4$, $a = 0.5$,
$\theta = 0.02$, $T = 1.0$.}
  \label{figure:PZDiscrepancyT1Y}
\end{figure}

\begin{figure}[htbp]
  \centering
 \includegraphics[width=0.8\textwidth]{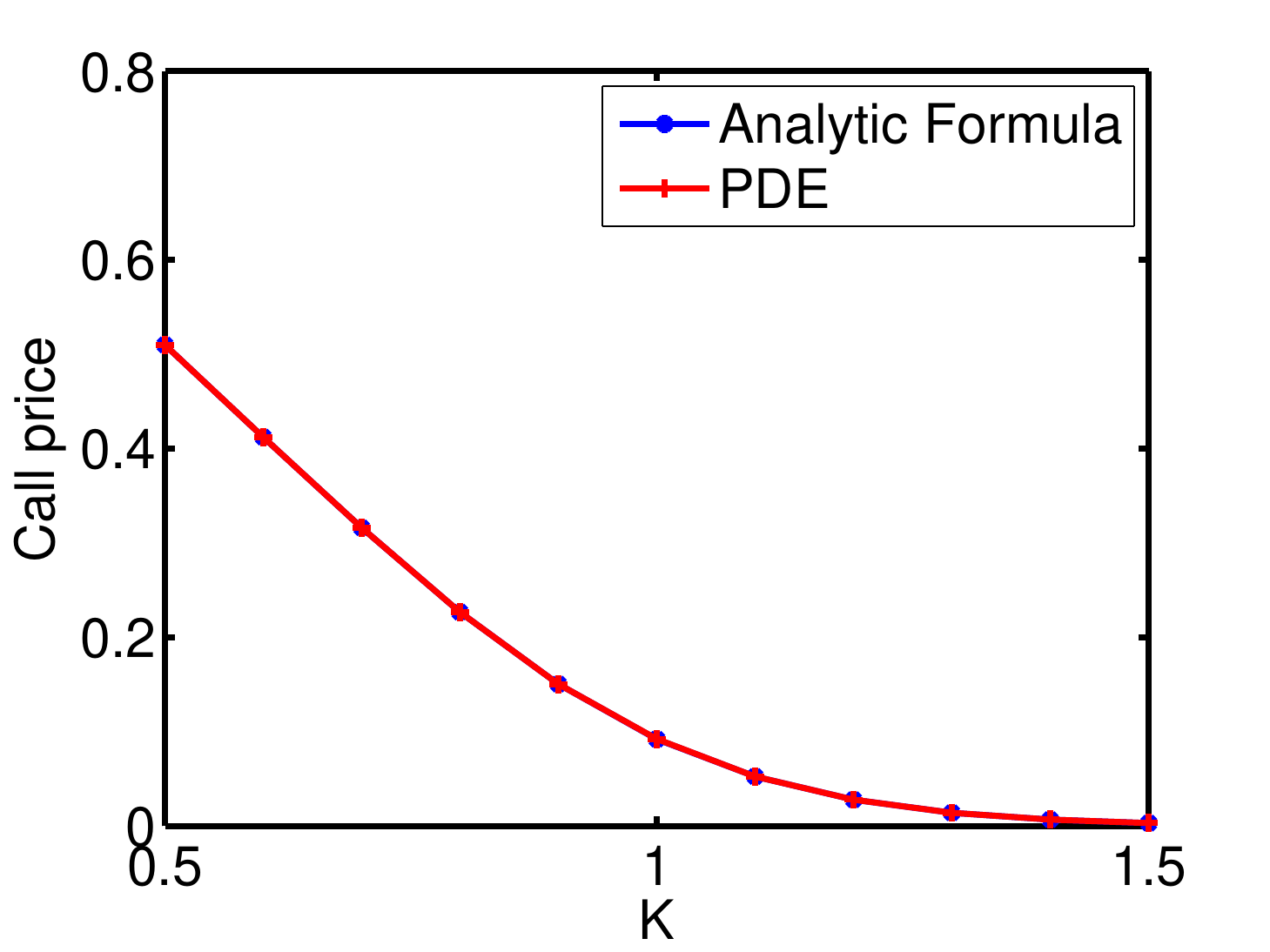}
  \caption{ Parameters $S_0 = 1.0$, 
$r_0 = 0.02$, $\sigma_1 = 0.2$, $\sigma_2 = 0.04$, $\rho = 0.4$, $a = 0.5$,
$\theta = 0.02$, $T = 1.0$.}
  \label{figure:CallPricesT1Y}
\end{figure}

\begin{figure}[htbp]
  \centering
  \includegraphics[width=0.8\textwidth]{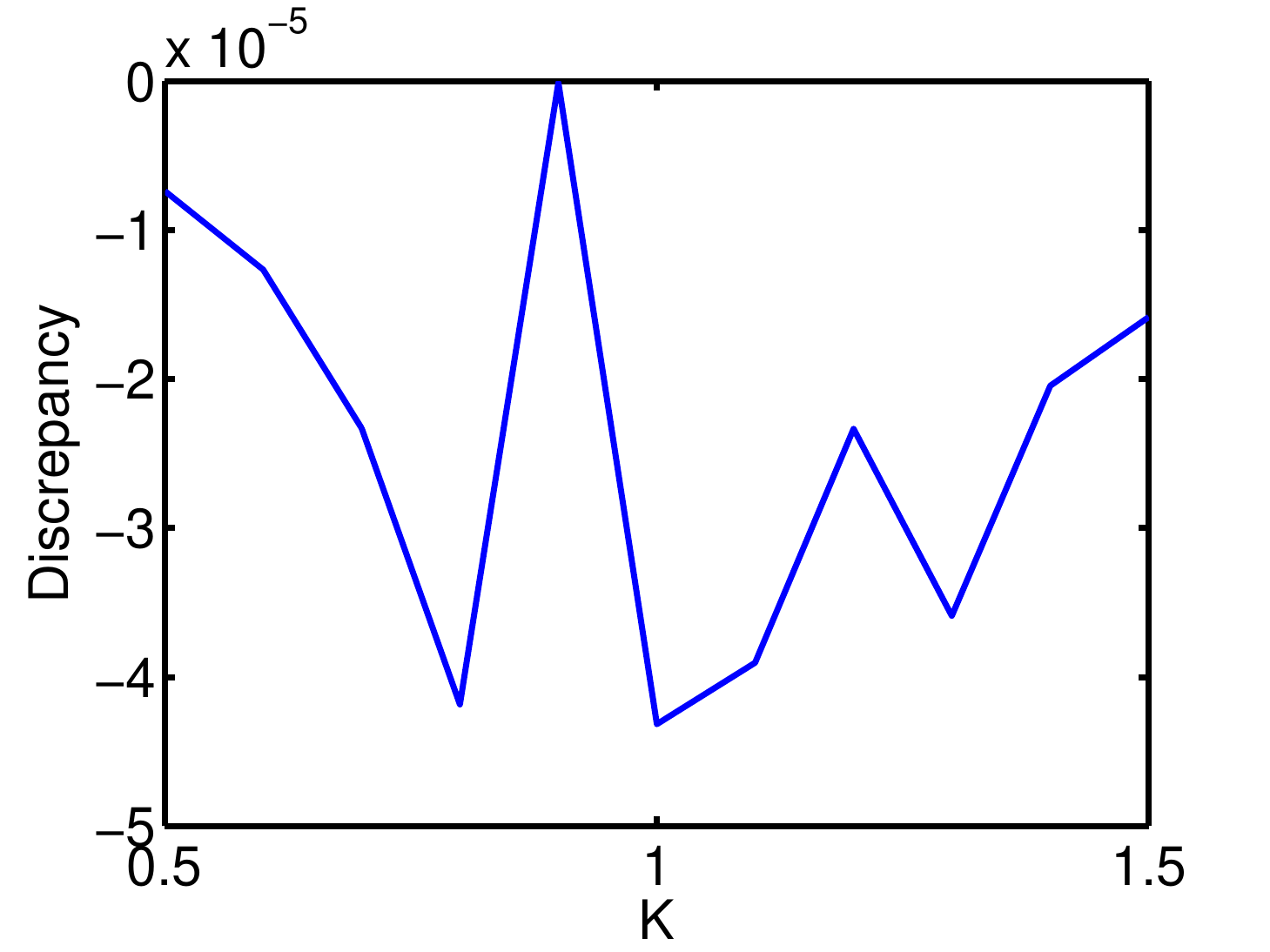}
  \caption{ Parameters $S_0 = 1.0$, 
$r_0 = 0.02$, $\sigma_1 = 0.2$, $\sigma_2 = 0.04$, $\rho = 0.4$, $a = 0.5$,
$\theta = 0.02$, $T = 1.0$.}
  \label{figure:PricesDiscrepancyT1Y}
\end{figure}

\begin{figure}[htbp]
  \centering
    \includegraphics[width=0.8\textwidth]{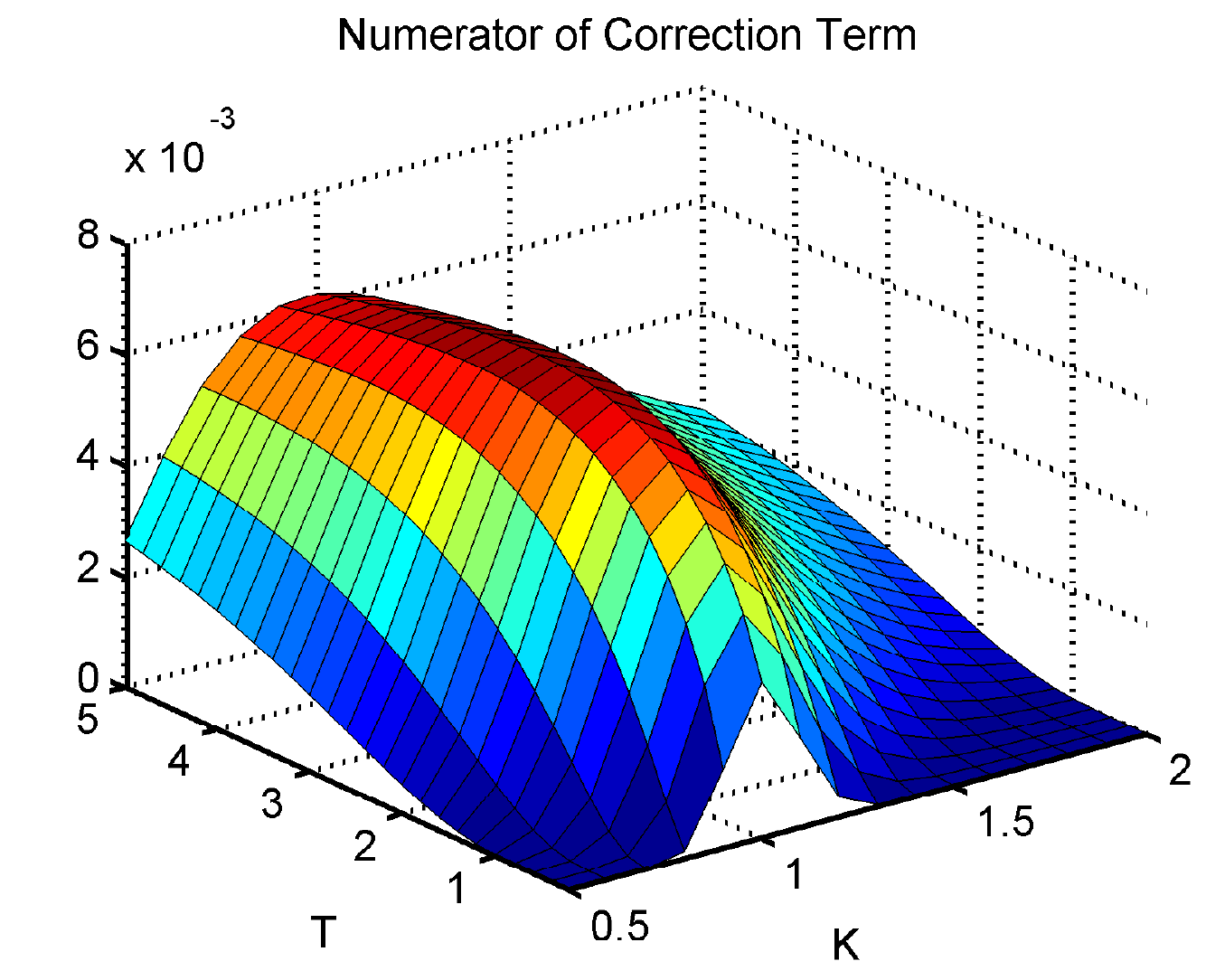}
  \caption{ Parameters $S_0 = 1.0$, 
$r_0 = 0.02$, $\sigma_1 = 0.2$, $\sigma_2 = 0.04$, $\rho = 0.4$, $a = 0.5$,
$\theta = 0.02$.}
  \label{figure:CorrectiveTermNumeratorSet1}
\end{figure}

\begin{figure}[htbp]
  \centering
\includegraphics[width=0.8\textwidth]{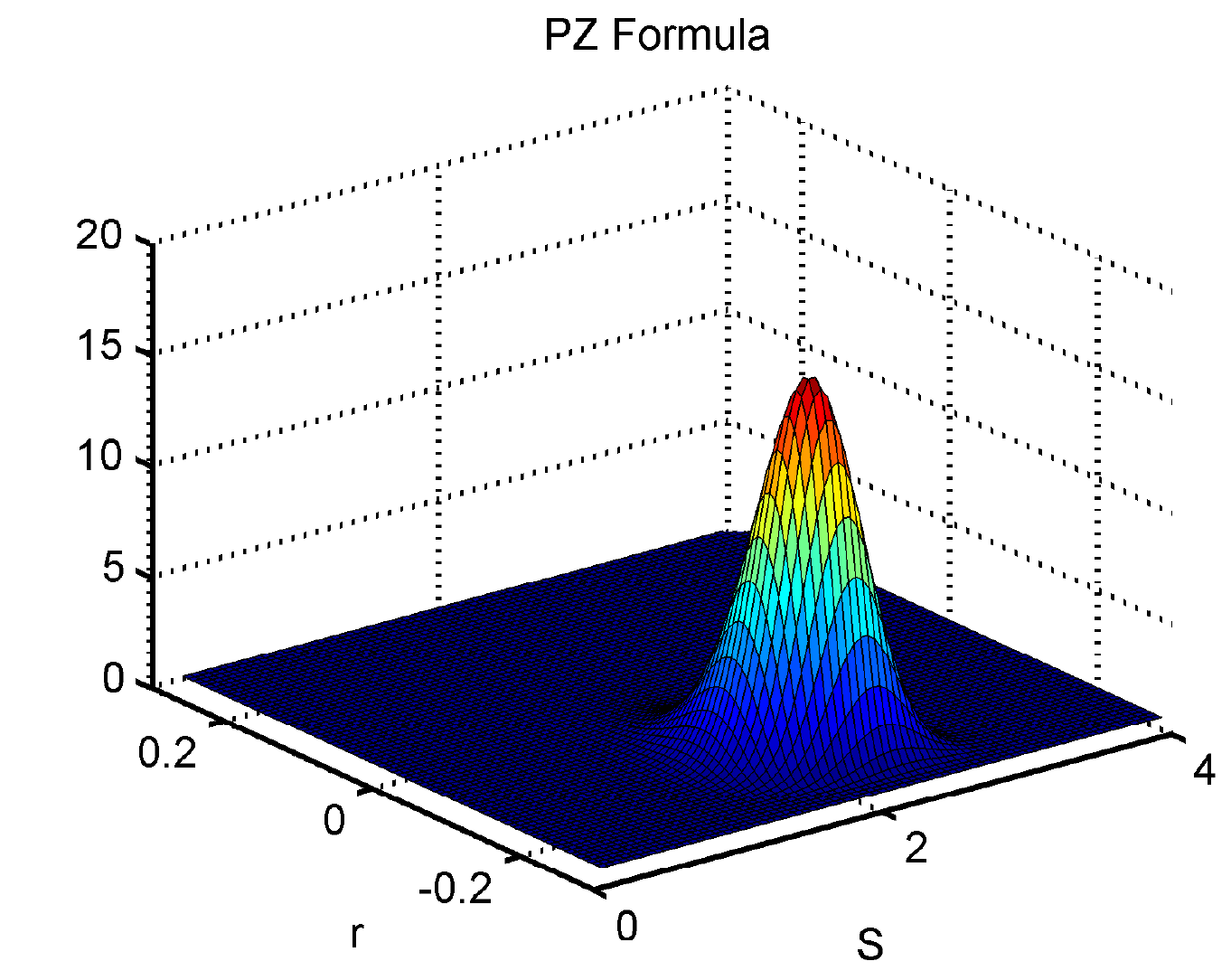}
  \caption{ Parameters $S_0 = 1.0$, 
$r_0 = 0.02$, $\sigma_1 = 0.2$, $\sigma_2 = 0.04$, $\rho = -0.4$, $a = 0.5$,
$\theta = 0.02$, $T = 2.0$.}
  \label{figure:PZFormulaT2Y}
\end{figure}

\begin{figure}[htbp]
  \centering
 \includegraphics[width=0.8\textwidth]{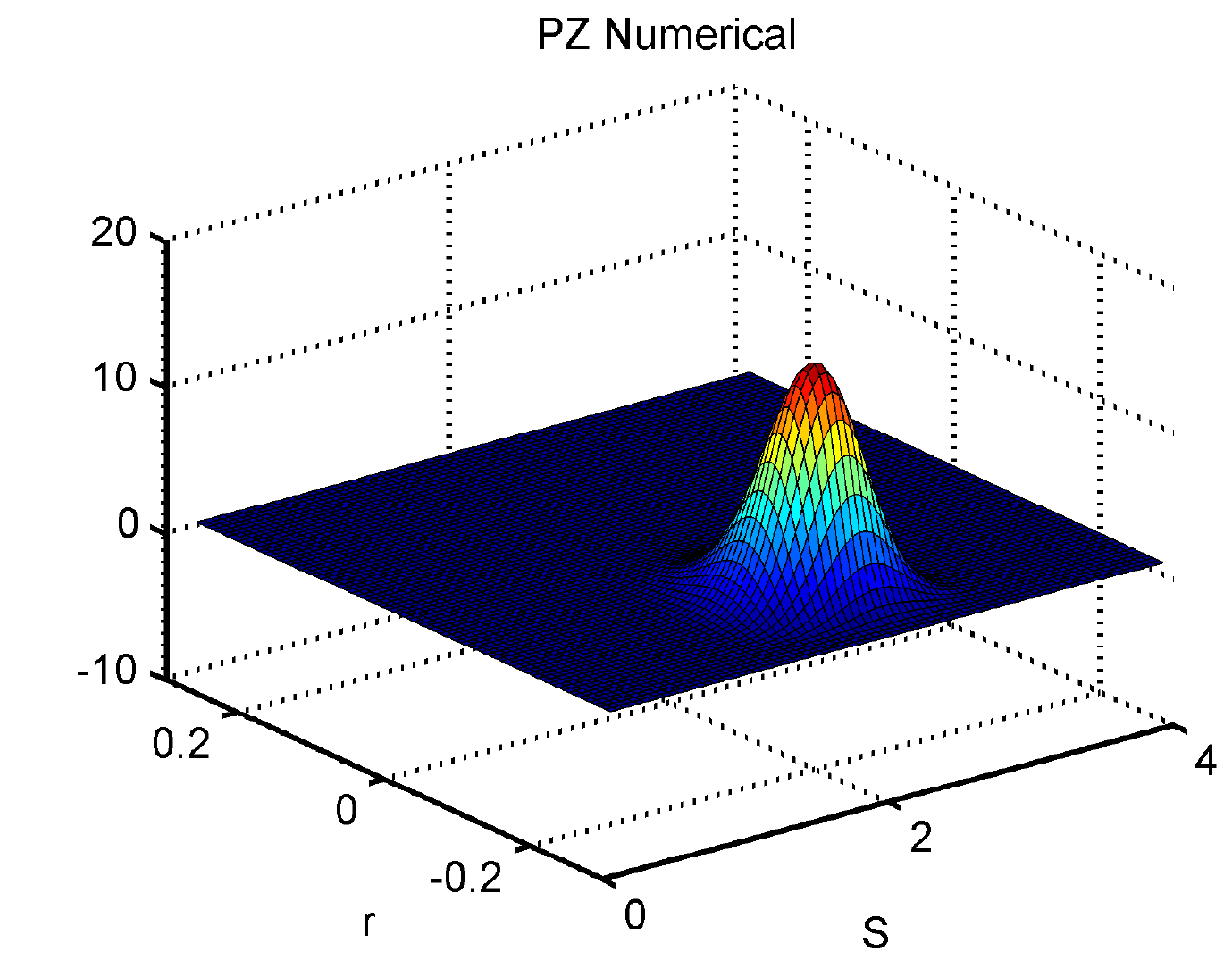}
  \caption{ Parameters $S_0 = 1.0$, 
$r_0 = 0.02$, $\sigma_1 = 0.2$, $\sigma_2 = 0.04$, $\rho = -0.4$, $a = 0.5$,
$\theta = 0.02$, $T = 2.0$.}
  \label{figure:PZNumericalT2Y}
\end{figure}

\begin{figure}[htbp]
  \centering
   \includegraphics[width=0.8\textwidth]{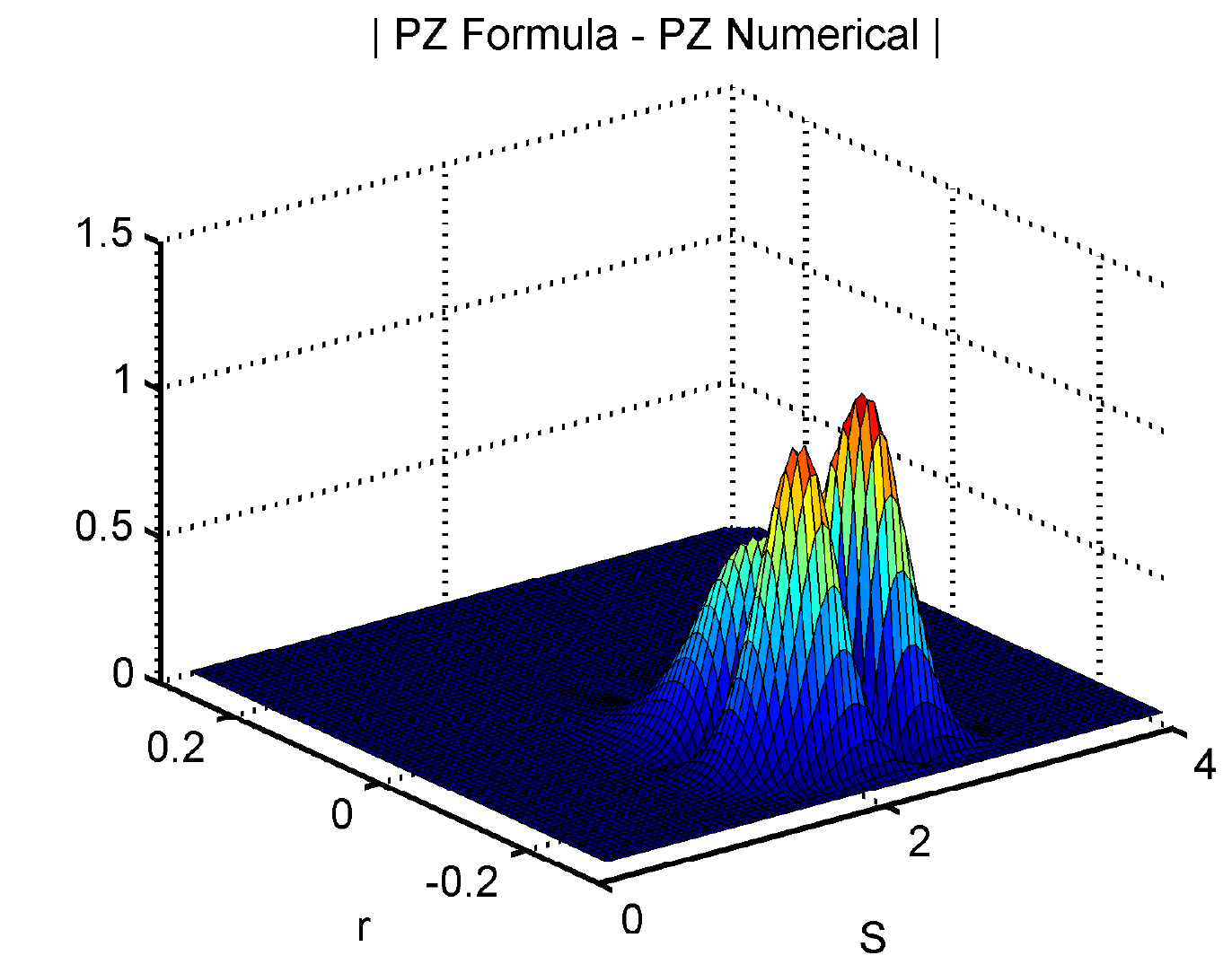}
  \caption{ Parameters $S_0 = 1.0$, 
$r_0 = 0.02$, $\sigma_1 = 0.2$, $\sigma_2 = 0.04$, $\rho = -0.4$, $a = 0.5$,
$\theta = 0.02$, $T = 2.0$.}
  \label{figure:PZDiscrepancyT2Y}
\end{figure}

\begin{figure}[htbp]
  \centering
 \includegraphics[width=0.8\textwidth]{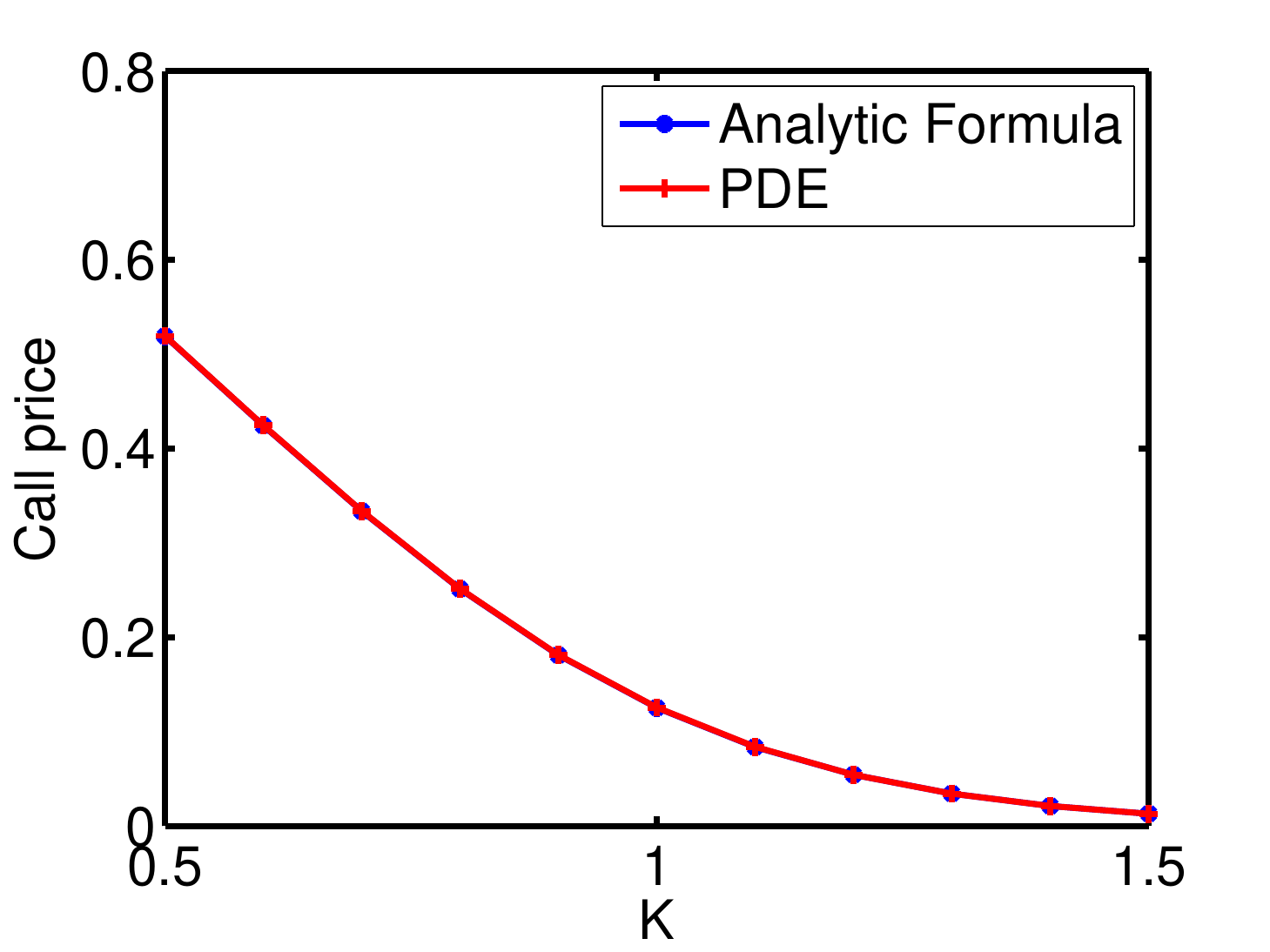}
  \caption{ Parameters $S_0 = 1.0$, 
$r_0 = 0.02$, $\sigma_1 = 0.2$, $\sigma_2 = 0.04$, $\rho = -0.4$, $a = 0.5$,
$\theta = 0.02$, $T = 2.0$.}
  \label{figure:CallPricesT2Y}
\end{figure}

\begin{figure}[htbp]
  \centering
  \includegraphics[width=0.8\textwidth]{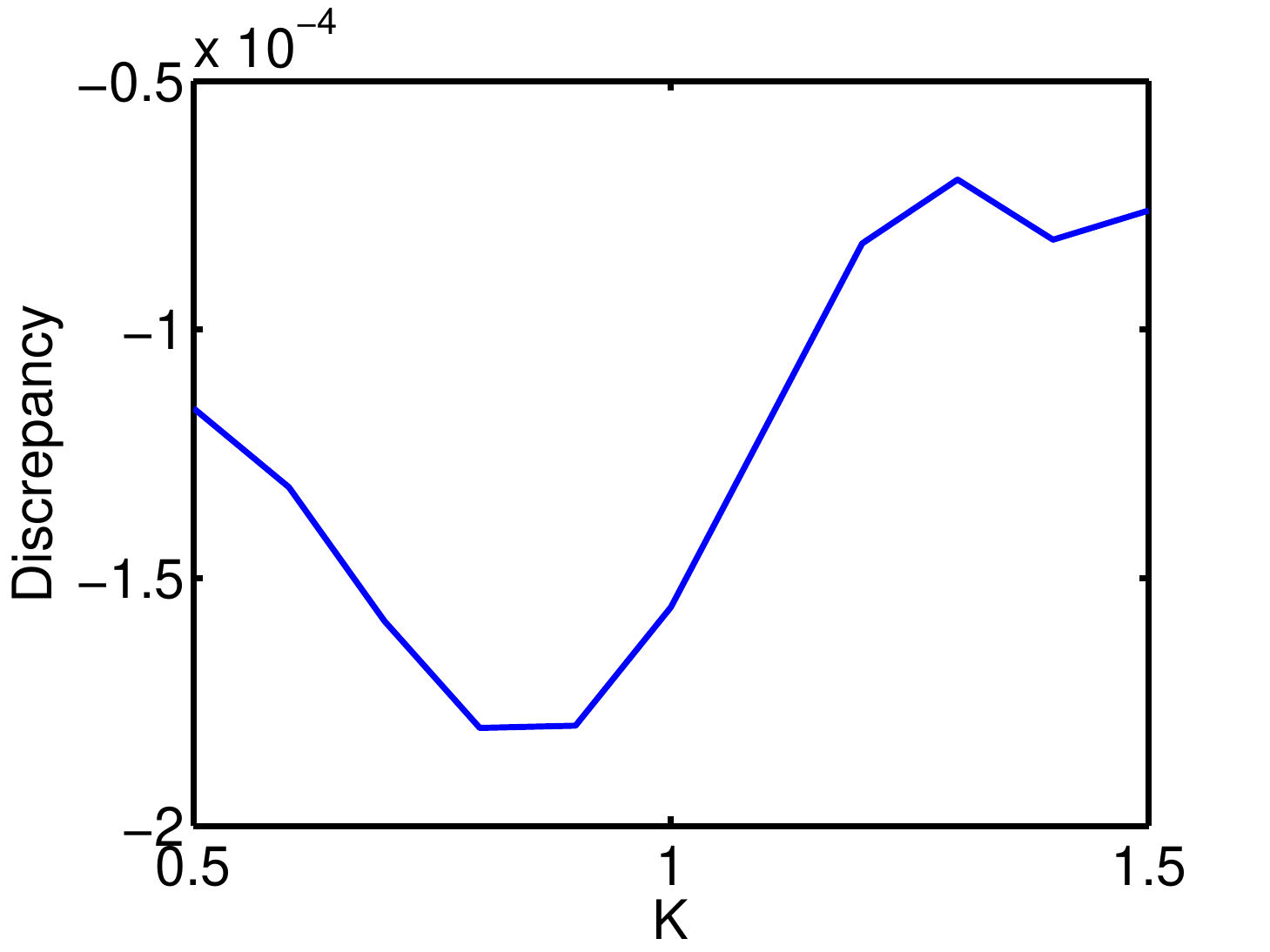}
  \caption{ Parameters $S_0 = 1.0$, 
$r_0 = 0.02$, $\sigma_1 = 0.2$, $\sigma_2 = 0.04$, $\rho = -0.4$, $a = 0.5$,
$\theta = 0.02$, $T = 2.0$.}
  \label{figure:PricesDiscrepancyT2Y}
\end{figure}

\begin{figure}[htbp]
  \centering
  \includegraphics[width=0.8\textwidth]{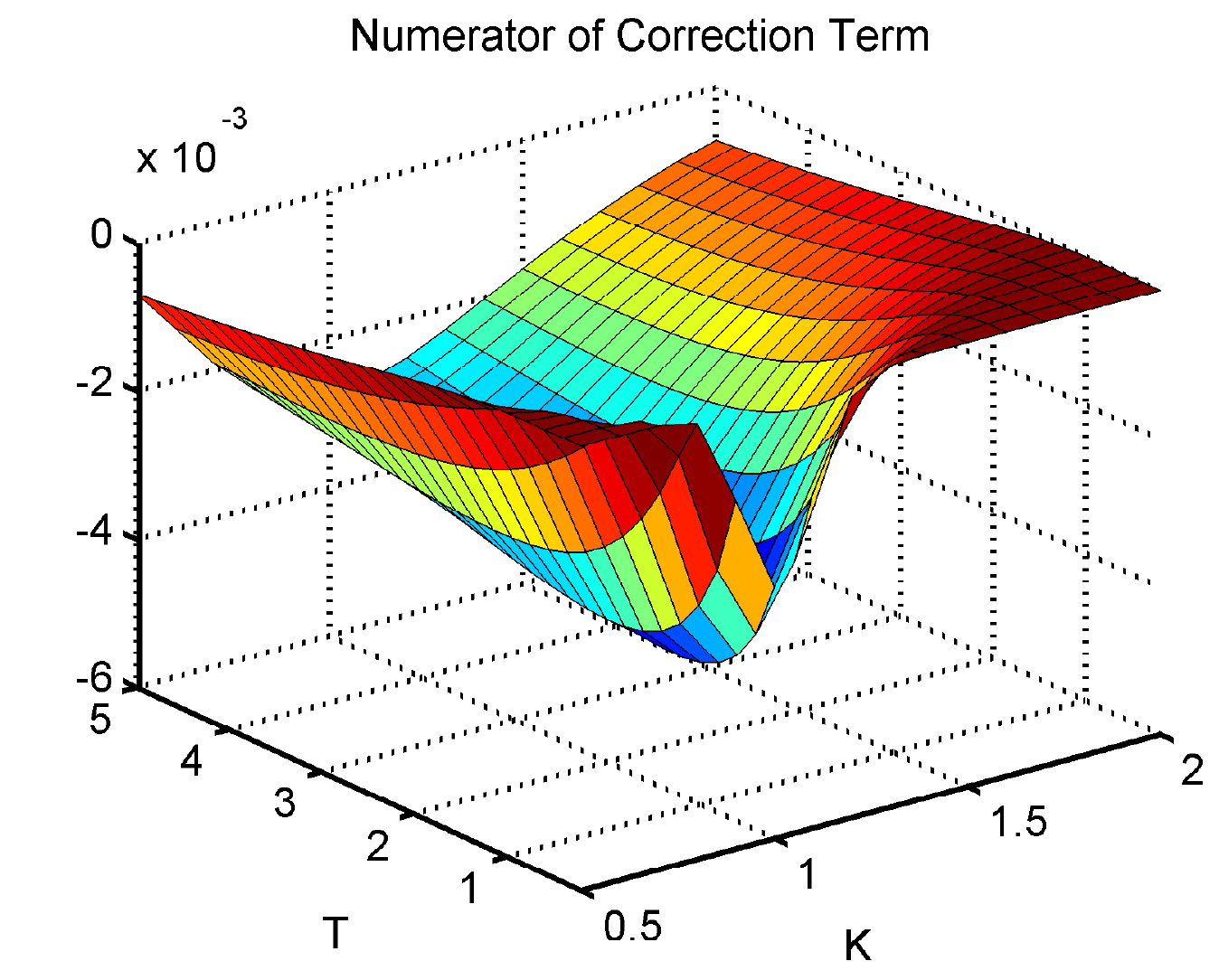}
  \caption{ Parameters $S_0 = 1.0$, 
$r_0 = 0.02$, $\sigma_1 = 0.2$, $\sigma_2 = 0.04$, $\rho = -0.4$, $a = 0.5$,
$\theta = 0.02$.}
  \label{figure:CorrectiveTermNumeratorSet2}
\end{figure}

\begin{figure}[htbp]
  \centering
 \includegraphics[width=0.8\textwidth]{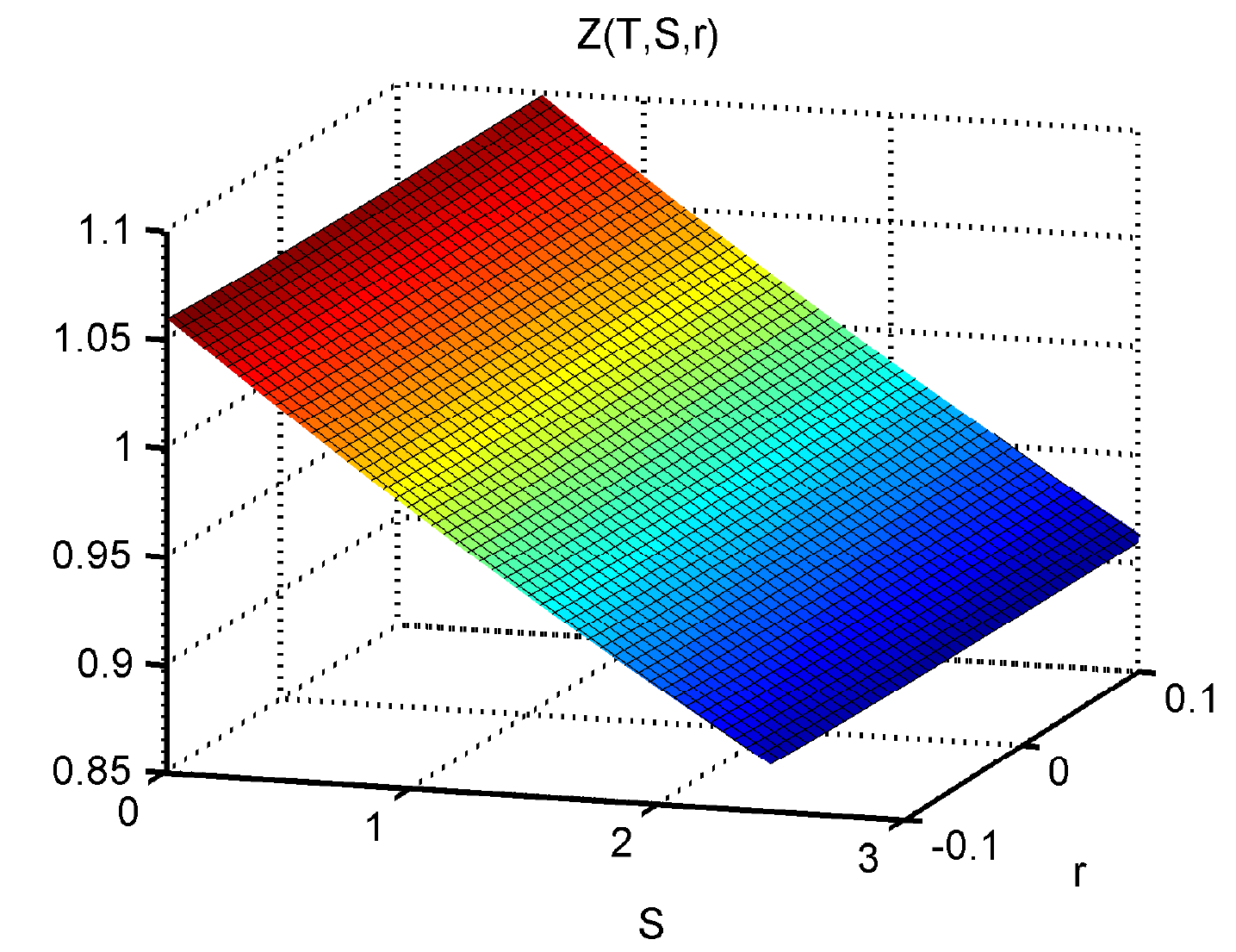}
  \caption{ Parameters $S_0 = 1.0$, 
$r_0 = 0.02$, $\sigma_1 = 0.2$, $\sigma_2 = 0.04$, $\rho = 0.4$, $a = 0.5$,
$\theta = 0.02$, $T = 1.0$.}
  \label{figure:ZfunctionT1Y}
\end{figure}

\begin{figure}[htbp]
  \centering
 \includegraphics[width=0.8\textwidth]{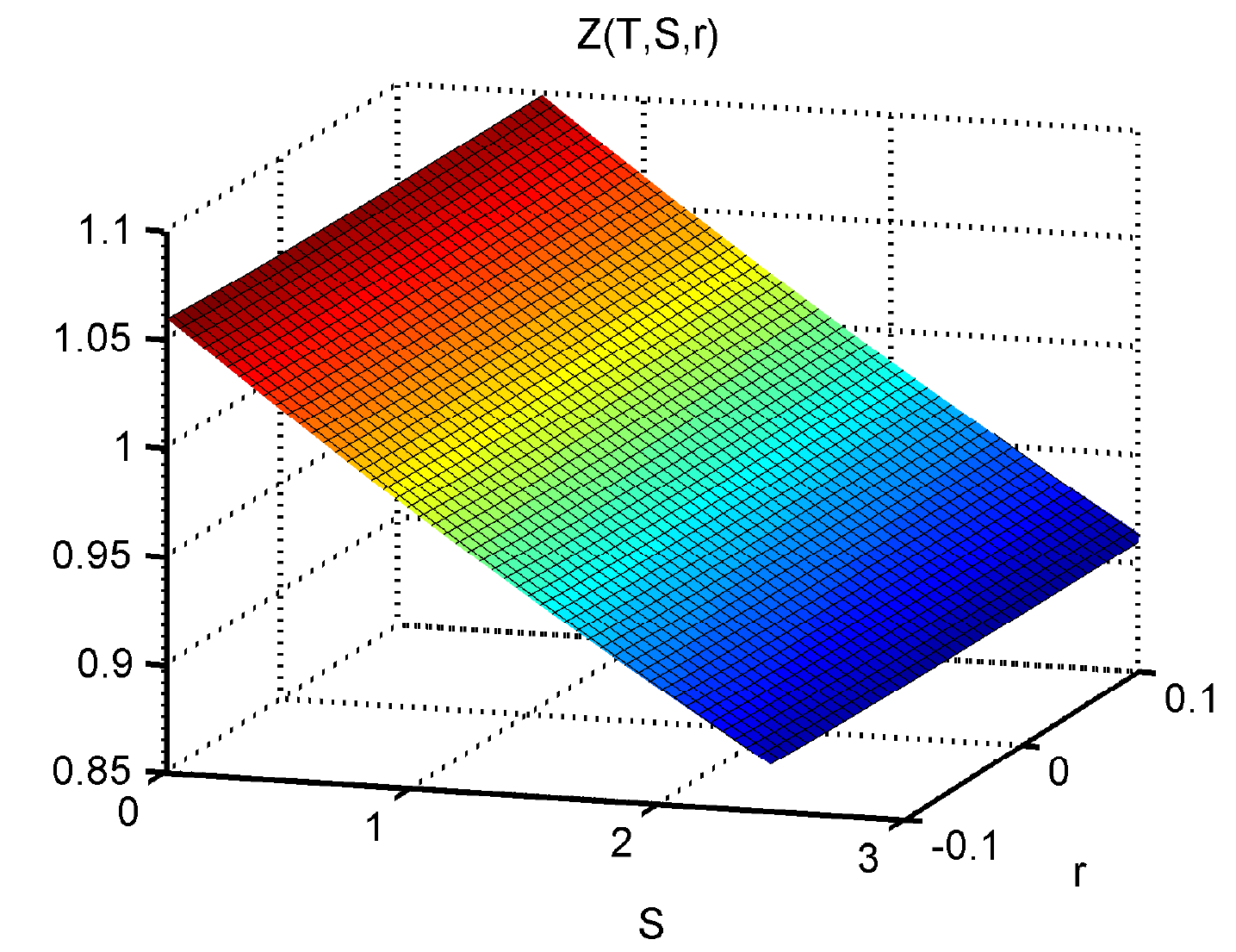}
  \caption{ Parameters $S_0 = 1.0$, 
$r_0 = 0.02$, $\sigma_1 = 0.2$, $\sigma_2 = 0.04$, $\rho = -0.4$, $a = 0.5$,
$\theta = 0.02$, $T = 2.0$.}
  \label{figure:ZfunctionT2Y}
\end{figure}

\subsection{Hyperbolic local volatility Hull-White model}

In our second example, we consider a skew model given by

\begin{equation} \label{BSVasicekSDEs}
\left\lbrace
\begin{array}{l}
	\frac{dS(t)}{S_t} = r(t)dt + \sigma_H(S_t) dW^1(t) , \hspace{0.3cm} S(0) = s_0 ,\\
	dr(t) = a(\theta(t)-r(t))dt +  \sigma_2 (\rho dW^1(t) + \sqrt{1-\rho^2} dW^2(t)), \hspace{0.3cm} r(0) = r_0, \\
	\end{array}
\right.
\end{equation}
where
\begin{equation}
\sigma_H(S_t) = \nu \Big\{ \frac{(1-\beta+\beta^2)}{\beta}  +\frac{(\beta-1)}{\beta S_t}  \big(\sqrt{S_t^2+\beta^2(1-S_t)^2}-\beta\big) \Big\},
\end{equation}
with $ \nu > 0$ presents the level of volatility and $ \beta \in [0, 1]$ shows the skew parameter.\\

This model introduced in \cite{Jackel10} behaves
closely to the CEV  model and has been used for
numerical experiments as in \cite{BompisHok14,HokPapLeg17}. It presents the advantage to avoid zero to be
an attainable boundary and then allows to avoid some numerical instabilities
as seen in the CEV model when the underlying $S$ is close to $0$ (see e.g
\cite{And00}). It corresponds to the Black-Scholes model for $ \beta = 1$ and exhibits a skew for the volatility surface when $ \beta \neq 1$.  Figure \ref{figure:HyperbolicLVImpactBeta} illustrates the impact of the parameter $\beta$ on the skew of the volatility surface. We observe that the skew increases significantly with decreasing value of $\beta$. For example with $\nu = 0.2, \, \beta = 0.2$, the difference in volatility between strikes at $50\%$ and at $100\%$ is about $20 \%$.\\

We run two family of tests by considering respectively negative correlation $\rho = -30 \%$ and  positive correlation $\rho = 30 \%$. Other model parameters were chosen in Table \ref{tab:HyperbolicHWModelParam}.
For both tests, we solve the PDE in equation (\ref{fwdeqpz}) for $PZ$ up to maturity $T$ in an uniform grid with $ds = 0.012, \, dr = 0.002, \, dt = 0.0099$. Then we perform the European call pricing with maturity $T$ for various strike by numerical integration using PDE grid points and solution. We also run a Monte Carlo simulation for doing pricing by using Euler discretisation for the SDE (\ref{BSVasicekSDEs}) with $dt = \frac{1}{300}$ and one million of paths (see e.g \cite{Glasserman03} or \cite{KloedenPlaten92}). \\

For set $1$ (respectively for set $2$) the pricing results are shown in Figure \ref{figure:CallPricesT1YNegativeRho} (respectively in Figure \ref{figure:CallPricesT1YPositiveRho}) and their discrepancies in Figure \ref{figure:PricesDiscrepancyT1YNegativeRho} (respectively in Figure \ref{figure:PricesDiscrepancyT1YPositiveRho}).  We obtain very accurate results as the differences of prices given by using PDE and Monte Carlo methods are all within a couple of basis points for all strikes in the range $[0, 2]$.

 \begin{center}
\begin{tabular}{|c|c|}
\hline
$S_0$ 	& 1\\
\hline
$r_0$ 	& $3.75 \%$\\
\hline
$\nu$ 	& $20 \%$ \\
\hline
$\beta$ 	& $0.5$ \\
\hline
$\sigma_2$ 	& $4 \%$\\
\hline
$a$ 	& $0.5$ \\
\hline
$T$ 	& $ 1 $ \\
\hline
\end{tabular}
\captionof{table}{Model parameters for the hyperbolic local volatility Hull-White model.}\label{tab:HyperbolicHWModelParam}
\end{center}

\begin{figure}[htbp]
  \centering
 \includegraphics[width=0.8\textwidth]{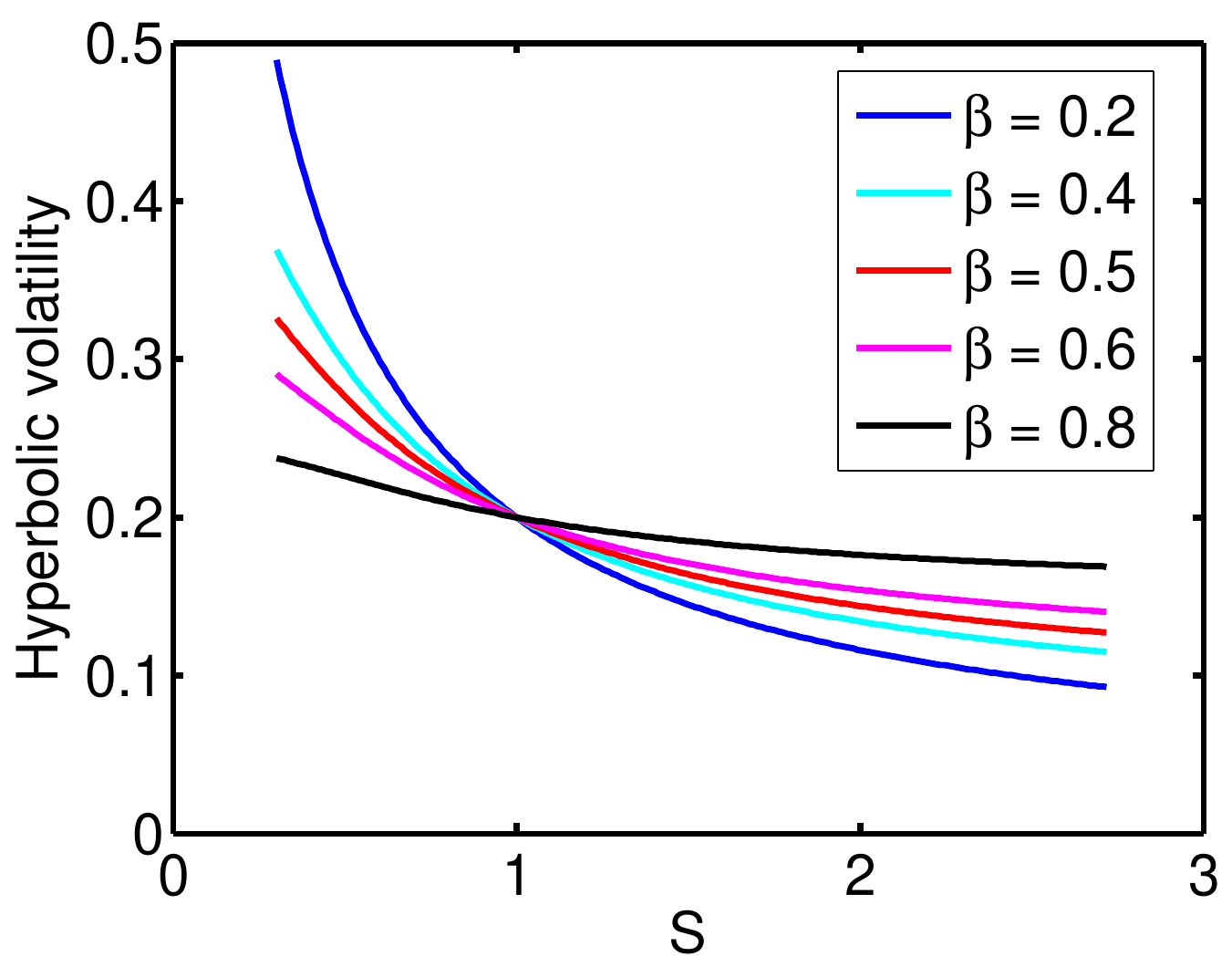}
  \caption{ Impact of the value $\beta$ on the hyperbolic local volatility $\sigma_H$ 
	for fixed volatility level $\nu = 0.2$. }
  \label{figure:HyperbolicLVImpactBeta}
\end{figure}

\begin{figure}[htbp]
  \centering
 \includegraphics[width=0.8\textwidth]{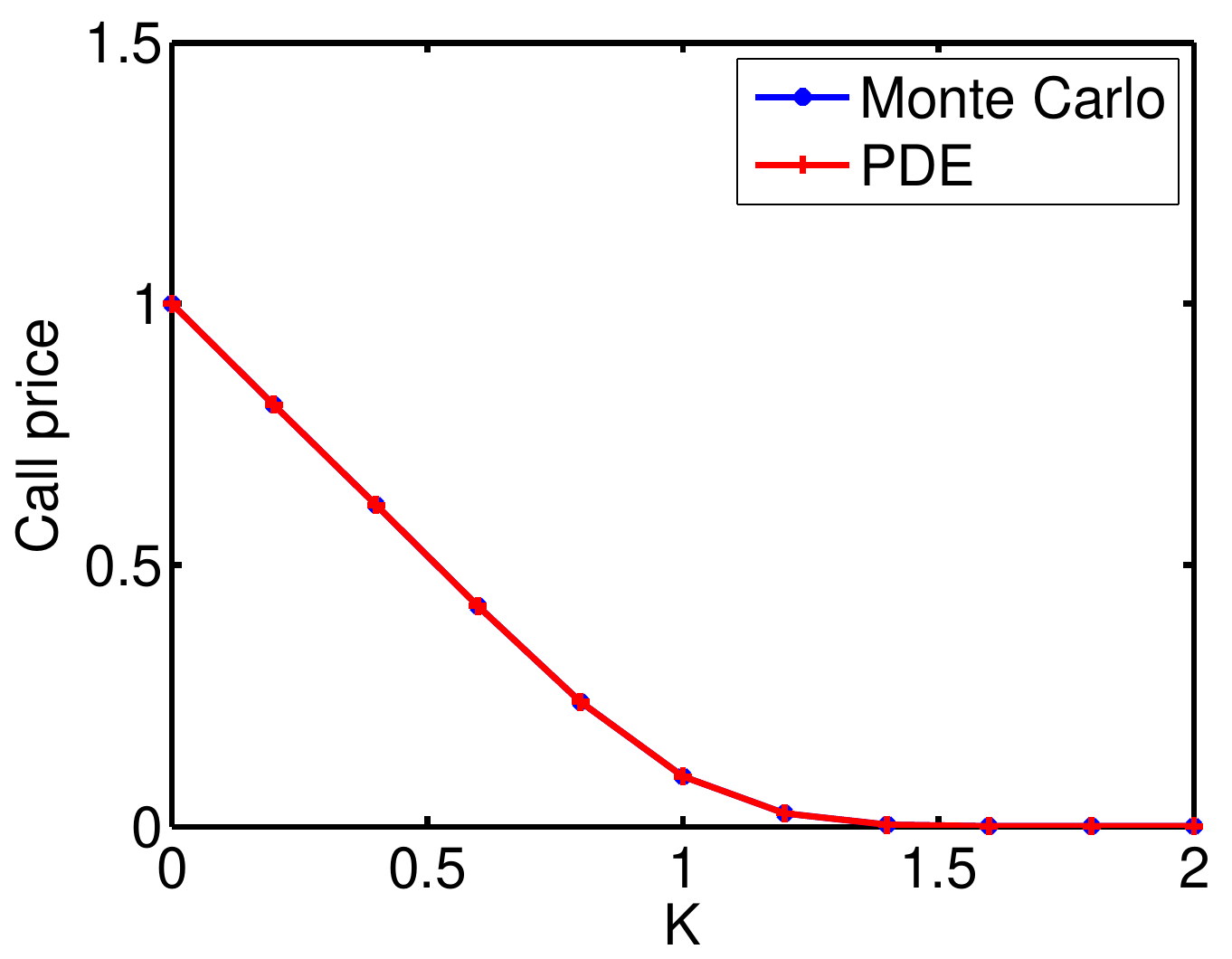}
  \caption{ Parameters $S_0 = 1.0$, 
$r_0 = 3.75 \%$, $\nu = 20\%$, $\beta = 0.5$ , $\sigma_2 = 4 \%$, $\rho = -0.3$, $a = 0.5$, $T = 1.0$.}
  \label{figure:CallPricesT1YNegativeRho}
\end{figure}

\begin{figure}[htbp]
  \centering
  \includegraphics[width=0.8\textwidth]{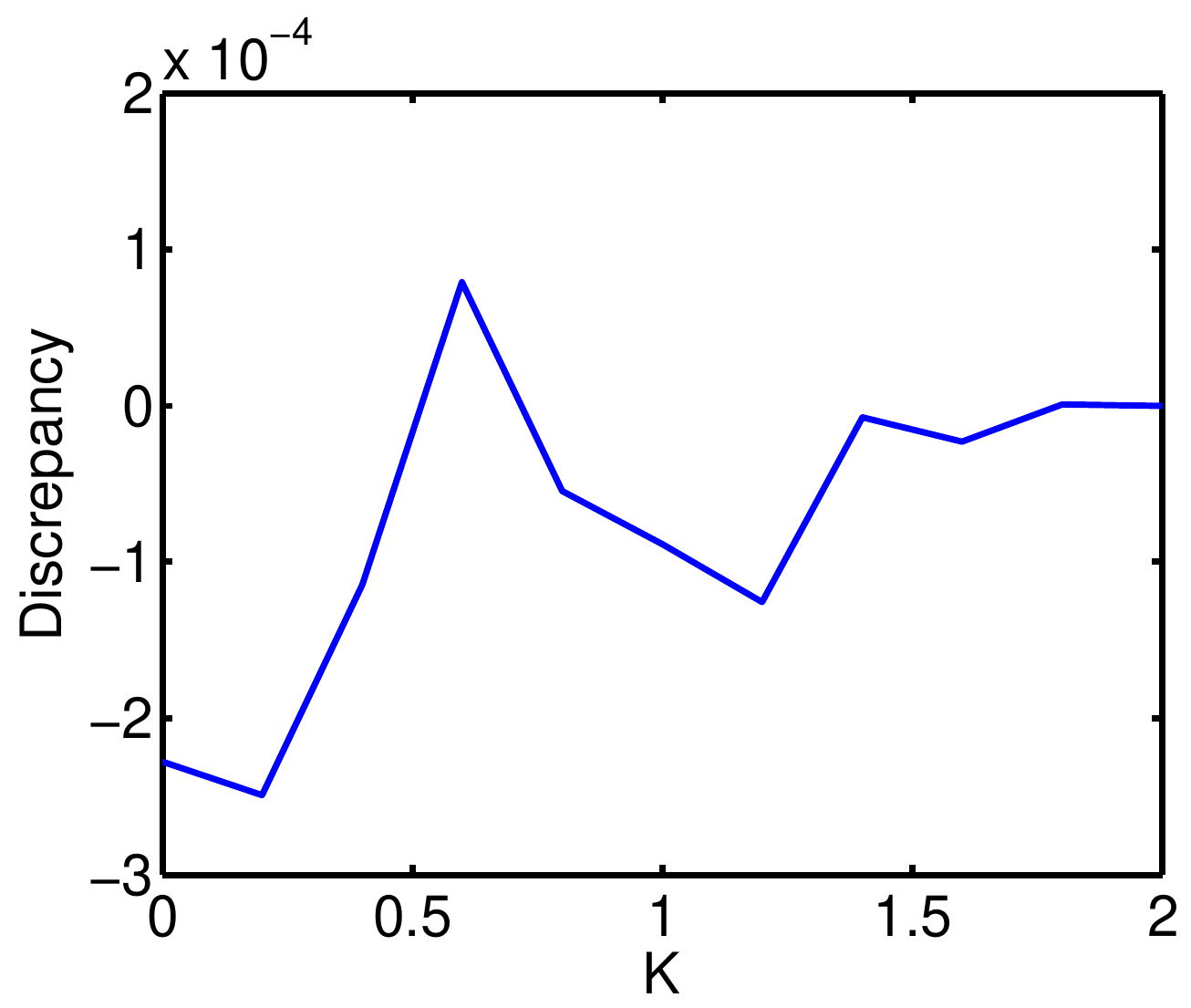}
  \caption{ Parameters $S_0 = 1.0$, 
$r_0 = 3.75 \%$, $\nu = 20\%$, $\beta = 0.5$ , $\sigma_2 = 4 \%$, $\rho = -0.3$, $a = 0.5$, $T = 1.0$.}
  \label{figure:PricesDiscrepancyT1YNegativeRho}
\end{figure}

\begin{figure}[htbp]
  \centering
   \includegraphics[width=0.8\textwidth]{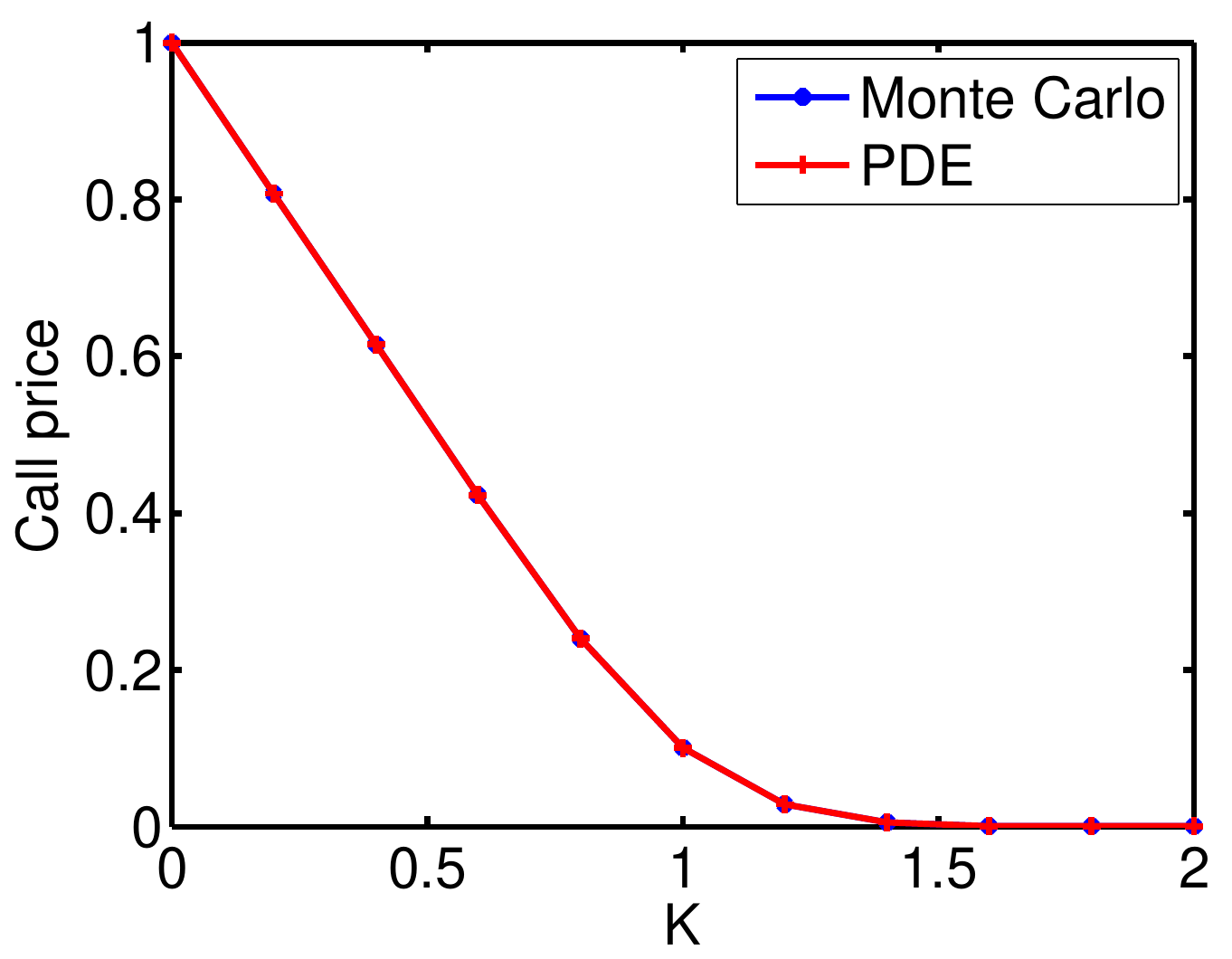}
  \caption{ Parameters $S_0 = 1.0$, 
$r_0 = 3.75 \%$, $\nu = 20\%$, $\beta = 0.5$ , $\sigma_2 = 4 \%$, $\rho = 0.3$, $a = 0.5$, $T = 1.0$.}
  \label{figure:CallPricesT1YPositiveRho}
\end{figure}

\begin{figure}[htbp]
  \centering
\includegraphics[width=0.8\textwidth]{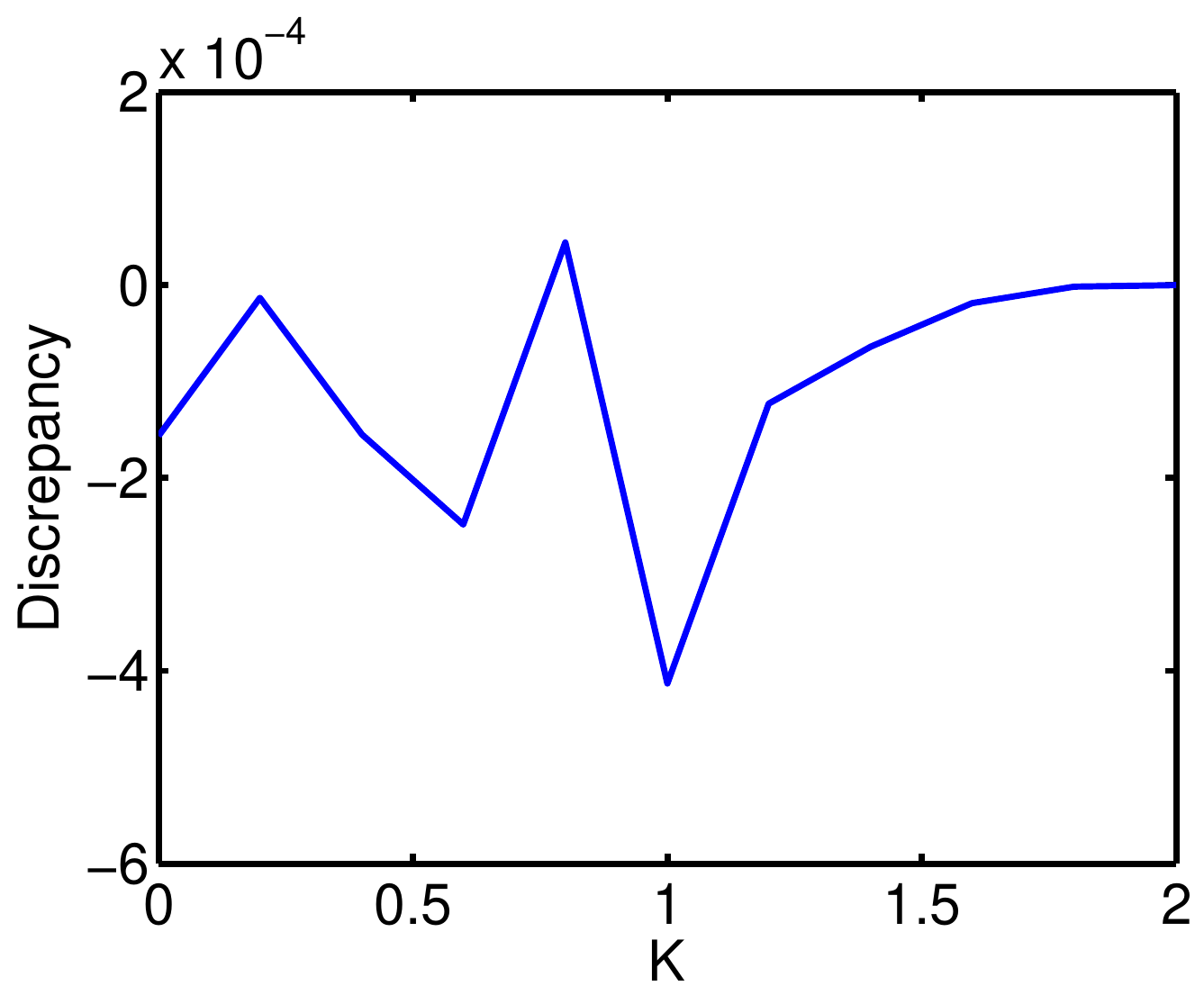}
  \caption{ Parameters $S_0 = 1.0$, 
$r_0 = 3.75 \%$, $\nu = 20\%$, $\beta = 0.5$ , $\sigma_2 = 4 \%$, $\rho = 0.3$, $a = 0.5$, $T = 1.0$.}
  \label{figure:PricesDiscrepancyT1YPositiveRho}
\end{figure}

\section{Conclusions and discussions}

In this paper, we have proposed a new PDE-based technique for calibration on the local volatility model with stochastic interest rate. The main results are the derivation of the forward equation satisfied by $P(t, S, r)Z(t, S, r)$ and the constructed PDE solver based on ADI scheme 
which leads to a more efficient calibration method. Besides, some techniques of accelerating the calibration algorithm are also introduced which make the model practical for real time execution.
The numerical experiments complement our theoretical analysis and show consistent tests results.
Furthermore, the discussed model is actually a general case which can cover most of the well-known extension of
 local volatility models used these days. Therefore this calibration framework is useful for 
 various problems in different asset class markets like equity, exchange or inflation.  \\

We suggest a couple of interesting avenues of research:

\begin{itemize}

\item Here, we have focused our numerical experiments on the resolution of the forward equation (\ref{fwdeqpz}) with two widely used hybrid models in quantitative finance: The Black-Scholes Hull-White 
and the Hyperbolic Local Volatility Hull-White models. It would be interesting to complement the testing of the PDE calibration procedure with live market data.

\item Through the numerical tests, our simple ADI scheme gives good convergence results and shows robustness w.r.t strong skew and high correlations parameters. 
Performing numerical analysis and providing comparisons w.r.t modern ADI schemes are parts of our future research.

\end{itemize}

\section{Appendix}

For the proof of corollary (\ref{BSStochasticIRGreeks}), we provide the following useful lemma

{
\lemma{ Using the definitions in proposition (\ref{BlackScholesStochasticRates})
 and corollary (\ref{BSStochasticIRGreeks}), we have

\begin{equation}\label{usefulRel}
S_0n(d_1) = K ZC(0, T) n(d_2)
\end{equation}
}
}

{
\proof{

\begin{align}
d_2^2 - d_1^2 &= (d_2-d_1)(d_2+d_1)\\
  			  &= -\sqrt{g(T)} (d_2+d_1)\\
  			  &= -\sqrt{g(T)} \left( 2d_1 - \sqrt{g(T)} \right)\\
  			  &= -2 \left( \log \left( \frac{S_0}{K} \right) - \log ZC(0,T) \right)\\
\log \left( \frac{n(d_1)}{n(d_2)} \right) &= \log \left( \frac{K\log ZC(0,T)}{S_0} \right)  &  
\end{align}

From the last expression, we deduce directly (\ref{usefulRel}). \\

}
}

For expression (\ref{GreekCT}), we write 

\begin{align}
C_T(T, K) &= S_0 n(d_1) d_{1, T} - K \left[ ZC_T(0,T) N(d_2) + ZC(0,T) n(d_2) d_{2,T}  \right]\\
		  &= S_0 n(d_1)( d_{1, T}-d_{2, T} ) - K ZC_T(0,T) N(d_2) \\
		  &= \frac{S_0 n(d_1)}{2} \frac{\hat{\sigma}^2(T)}{\sqrt{g(T)}}  + K ZC(0, T)f(0, T) N(d_2)
\end{align}

where we have used (\ref{usefulRel}) in the second equality, 
$d_{1, T}-d_{2, T} = \frac{1}{2}\frac{\hat{\sigma}^2(T)}{\sqrt{g(T)}}$ and 
$ZC_T(0,T) = -ZC(0,T) f(0,T)$ to obtain the third expression.

\begin{equation}
C_K(T, K) = S_0 n(d_1) d_{1, K} - ZC(0,T) \left[ N(d_2) + K n(d_2) d_{2, K} \right] 
\end{equation}

Using $d_{1, K} = d_{2, K}$ and result of lemma (\ref{usefulRel}), we get expression (\ref{GreekCK}). 
Finally, we obtain formula (\ref{GreekCKK}) by deriving (\ref{GreekCK}) w.r.t $K$ and using 
$d_{1, K} = -\frac{1}{K \sqrt{g(T)}}$.

\bibliographystyle{plain}

\end{document}